\documentclass[acmsmall]{acmart}
\usepackage[utf8]{inputenc}
\usepackage{amsmath}
\usepackage{bm}
\usepackage{graphicx,subcaption}
\usepackage[ruled,vlined]{algorithm2e}
\usepackage{algorithmic}
\usepackage{comment}
\usepackage{bigstrut}
\usepackage{multirow}
\usepackage{array}
\usepackage{tabularx}
\usepackage{hhline}
\usepackage{pifont}

\newcommand{\specialcell}[2][c]{%
   \begin{tabular}[#1]{@{}l@{}}#2\end{tabular}}

\ifodd 1
\newcommand{\com}[1]{\textbf{\color{red} (COMMENT: #1)}}
\newcommand{\Qiulin}[1]{\textbf{\color{blue}(COM: #1)}}
\newcommand{\mo}[1]{{\color{purple}#1}}
\else

\newcommand{\com}[1]{}
\newcommand{\mo}[1]{}
\newcommand{\Qiulin}[1]{}
\fi

\theoremstyle{plain}
\newtheorem{thm}{\protect\theoremname}
\theoremstyle{plain}

\theoremstyle{plain}
\newtheorem{prop}[thm]{\protect\propname}
\theoremstyle{plain}
\newtheorem{lem}[thm]{\protect\lemmaname}
\theoremstyle{plain}

\providecommand{\definitionname}{Definition}
\providecommand{\propname}{Proposition}
\providecommand{\conname}{Conjecture}
\providecommand{\lemmaname}{Lemma}
\providecommand{\theoremname}{Theorem}  

\newcommand{\probnameS}{\textsf{OOIC-S}}
\newcommand{\probname}{\textsf{OOIC-M}} 
\newcommand{\algname}{\textsf{A\&P}} 
\newcommand{\allallt}{\textsf{AAt}}
\newcommand{\crp}{\textsf{CR-Pursuit($\pi$)}}
\begin{document}

\title{Competitive Online Optimization with Multiple Inventories: A Divide-and-Conquer Approach}
\author{Qiulin Lin, Yanfang Mo, Junyan Su, and Minghua Chen}
\affiliation{%
 \department{School of Data Science}
 \institution{City University of Hong Kong}
 \city{Hong Kong}
  \country{China}
}

\begin{abstract}
 We study a competitive online optimization problem with multiple inventories. In the problem, an online decision maker seeks to optimize the allocation of multiple capacity-limited inventories over a slotted horizon, while the allocation constraints and revenue function come online at each slot.  The problem is challenging as we need to allocate limited inventories under adversarial revenue functions and  allocation constraints, while our decisions are coupled among multiple inventories and different slots. We propose a divide-and-conquer approach that allows us to decompose the problem into several single inventory problems and solve it in a two-step manner with almost no optimality loss in terms of competitive ratio (CR). Our approach  provides new angles, insights and results to the problem, which differs from the widely-adopted primal-and-dual framework. Specifically, when the gradients of the revenue functions are bounded in a positive range, we show that our approach can achieve a tight CR that is optimal when the number of inventories is small, which is better than all existing ones. For an arbitrary number of inventories, the CR we achieve is within an additive constant of one to a lower bound of the best possible CR among all online algorithms for the problem. We further characterize a general condition for generalizing our approach to different applications. For example, for a generalized one-way trading problem with price elasticity, where no previous results are available, our approach obtains an online algorithm that achieves the optimal CR up to a constant factor. 
\end{abstract}

\maketitle

\section{Introduction}\label{sec:intro}


Competitive online optimization is a fundamental tool for decision making with uncertainty. 
We have witnessed its wide applications spreading from EV charging~\cite{zhao2017robust,alinia2018competitive,lin2021minimizing,wang2018electrical}, micro-grid operations~\cite{lu2013online,zhang2016peak}, energy storage scheduling~\cite{mo2021eEnergy,mo2021infocom} to data center provisioning~\cite{lin2012dynamic,lu2012simple}, network optimization~\cite{shi2018competitive,guo2018joint}, and beyond. 
{Theoretically, there are multiple paradigms of general interest in the online optimization literature. Typical examples include the online covering and packing problem~\cite{buchbinder2009design}, online matching~\cite{TCS-allocation}, online knapsack packing~\cite{zhou2008budget}, one-way trading problem~\cite{el2001optimal}, and online optimization with switching cost~\cite{lin2012online,bansal20152}.} 


{Among these applications and paradigms, we focus on optimizing the allocation of limited resources, like inventories, budgets, or power, across a multi-round decision period with dynamic per-round revenues and allocation conditions.} For example, in online ads display platform (e.g., google search), the operator allocates display slots to multiple advertisers (\emph{cf.} inventories) with contracts on the maximum number (\emph{cf.} capacity) of impressions~\cite{Feldman2009Online}. At each round, the real-time search reveals the number of total available display slots and interesting slots for each advertiser (\emph{cf.} allocation constraint). It also reveals the payoff of each advertiser from impressions at the display slots (\emph{cf.} revenue function). The revenue of an advertiser relies on the number of obtained impressions and the quality of each impression relating to dynamic factors like the click rate and engagement rate~\cite{Devanur2012Concave}. 


The above observations motivate us to study the important paradigm~--~competitive online optimization problem under multiple inventories (\probname), where a decision-maker with multiple inventories of fixed capacities seeks to maximize the per-round separable revenue function by optimizing the inventory allocation at each round. The decision maker further faces two allocation constraints at each round, the allowance constraint that limits the total allocation of different inventories at the round and the rate constraint that limits the allocation of each inventory at the round. 
 The problem has two main challenges. First, as in online optimization under (single) inventory constraints ~\cite{cr_pursuit2018,Sun2020Competitive}, the decision-maker does not have access to future revenue functions, while the limited capacity of each inventory coupling the online decisions regarding each inventory across time. 
Second, the allocation constraints that couples the allocation decisions across the multiple inventories at each round.  The combination of allowance constraint and rate limit constraint appears frequently and is with known challenges in online matching and allocation problems~\cite{azar2006maximizing,KALYANASUNDARAM20003bmatching,TCS-allocation,Deng2018Capacity}. 

In the literature, the authors in~\cite{ma2020algorithms,Sun2020Competitive} tackle the problem using the well-established online primal-and-dual framework~\cite{buchbinder2009design,Devanur2013Randomized,niazadeh2013unified}. They design a threshold function for each inventory with regard to the allocated amount, which can be view as the marginal cost of the inventory. They then greedily allocate the inventory at each round by maximizing the pseudo-revenue function defined by the difference between the revenue function and the threshold function. In contrast, we propose, in this paper, a divide-and-conquer approach to the online optimization with multiple inventories. Our approach is novel and provides additional insights to the problem. It allows us to separate the two challenges of the problem, 1) the online allocation for each inventory subject to the limited inventory capacity and unknown future revenue functions, and 2) the coupled allocation among multiple inventories due to the allowance constraint at each round. 
In the following, we summarize our contributions.

First,  in Sec.~\ref{sec:single}, we generalize the \crp~algorithm~\cite{cr_pursuit2018} to tackle the single inventory case, \probnameS, which is an important component in our divide-and-conquer approach. We show that it achieves the optimal competitive ratio (CR) among all online algorithm for \probnameS. 

Second, in Sec.~\ref{sec:online}, we propose a divide-and-conquer approach to design online algorithms for online optimization under multiple inventories and dynamic allocation constraints. By adding an allowance allocation step, we decompose the multiple inventory problem into several single inventory problems with a small optimality loss in terms of CR.   
We show that when the revenue functions have bounded gradients, the CR achieved by our algorithm is optimal at a small number of inventories and within an additive constant to the lower bound of the best possible CR among all online algorithms for the problem at an arbitrary number of inventories. 

Third, in Sec.~\ref{sec:discussion-general-application}, we discuss generalizations of our proposed approach to broader classes of revenue functions. We provide a sufficient condition for applying our online algorithm and derive the corresponding CR it can achieve. For example, we consider revenue functions capturing the one-way trading problem with price elasticity, where only the results on single inventory case are available in existing literature~\cite{Sun2020Competitive,cr_pursuit2018}. We show that our approach obtains an online algorithm that achieves the optimal CR up to a constant factor. 

Finally, our results in Sec.~\ref{sec:Step-I} generalize the online allocation maximization problem in~\cite{azar2006maximizing} and the online allocation with free disposal problem in~\cite{Feldman2009Online} by admitting allowance augmentation in online algorithms, which is of independent interest. We show that we can improve the CR from~$\frac{e}{e-1}$ to~$\frac{1}{\pi}\cdot \frac{e^{\frac{1}{\pi}}}{e^{\frac{1}{\pi}}-1}$, when our online algorithms are endowed with~$\pi$-time augmentation in allowance and allocation rate at each round. 


\section{Related Work}\label{sec:related}


\begin{table}[tb]
	\begin{tabular}{l|l|lr|l|}
		\cline{2-5}
		\multicolumn{1}{c|}{} & Model & \multicolumn{2}{c|}{Existing results} & Our result$^\mathsection$ \\ \hhline{-====}
		\multicolumn{1}{|l||}{\multirow{2}{*}[-2pt]{\rotatebox[origin=c]{90}{Single}}} & Linear & $\ln\theta + 1^\star$ & \cite{lorenz2009optimal,Sun2020Competitive,el2001optimal,cr_pursuit2018} & $\ln\theta + 1^\star$ \bigstrut \\ \cline{2-5} 
		\multicolumn{1}{|l||}{} & Concave$^\dagger$ & $\ln \theta + 1^\star$ & \cite{Sun2020Competitive} & $\ln\theta + 1^\star$ \bigstrut \\ \hhline{=====}
		\multicolumn{1}{|l||}{\multirow{4}{*}{\rotatebox[origin=c]{90}{Multi}}} & $\theta=1$ & $e/(e-1)^\star$ & \cite{karp1990optimal,Devanur2013Randomized,KALYANASUNDARAM20003bmatching} & $e/(e-1)^\star$ \bigstrut \\ \cline{2-5} 
		\multicolumn{1}{|l||}{} & \multirow{2}{*}{Concave$^\dagger$} & \multirow{2}{*}{$<\ln\theta+2$} & \multirow{2}{*}{~\cite{Sun2020Competitive,ma2020algorithms}} & \multirow{1}{*}{$<\ln \theta + 2$, if $\ln\theta +1 < N$  \bigstrut} \\ 
		\multicolumn{1}{|l||}{} &  &  &  & {$\ln \theta + 1^\star$, otherwise  \bigstrut}\\\cline{2-5} 
		\multicolumn{1}{|l||}{} & \multirow{1}{*}{Concave$^\ddagger$} & \multirow{1}{*}{None} & \multirow{1}{*}{} & \multirow{1}{*}{$\Omega({\ln \theta})$ \bigstrut} \\ \hline
		\end{tabular}
	\\\footnotesize{
	    $\theta$ is a parameter representing the level of the input uncertainty.  $N$ is the number of inventories.\\ 
	    $\mathsection$: We propose a divide-and-conquer approach while existing work mostly base on primal-and-dual framework~\cite{buchbinder2009design}.\\
		$^\dagger$: Concave revenue function with bounded gradients. 
		$^\ddagger$: Concave revenue function with price elasticity (See Sec.~\ref{sec:discussion-general-application}).\\
		$^\star$: The result is optimal. 
		}
  	\caption{Comparison of our work and existing studies for online optimization under inventory constraints.}
  	\label{tab:related}%
\end{table}

We focus on the competitive online optimization problem with multiple inventories and dynamic allocation constraints. Our problem generalizes a couple of well-studied online problems, including the one-way trading problem~\cite{el2001optimal}, the online optimization under a single inventory constraint~\cite{cr_pursuit2018,Sun2020Competitive}, and the online fractional matching problem~\cite{KALYANASUNDARAM20003bmatching,azar2006maximizing}. Our results also reproduce the optimal CR under such settings.  Our problem is studied in ~\cite{Sun2020Competitive}, and a discrete counterpart is studied in~\cite{ma2020algorithms}. Compared with the existing study, we propose a novel divide-and-conquer approach. We show that our approach achieves a close-to-optimal CR, which notably matches the lower bound when the number of inventories is relatively small.  We summarize the related literature on online optimization under inventory constraints in Table~\ref{tab:related}. 

Our problem also covers a fractional version of online ads display problem~\cite{TCS-allocation}, which is an online matching problem with vertex capacity and edge value. No positive result is possible when the value is unbounded~\cite{Feldman2009Online}. In~\cite{Feldman2009Online}, the authors consider a model of ``free disposal'', i.e., the online decision maker can remove the past allocated edge without any cost (but can not re-choose the past edge). Here, we instead consider the case that the values of all edges are bounded in a pre-known positive range and no removable on past decision is available. We are interested in how the CR of an online algorithm behaviors with regard to the uncertain range of value. Interestingly, by our divide-and-conquer approach, we can extend the results in the ``free disposal'' model to the irremovable setting. Also, we provide additional insight and results to the problem when considering an online augmentation scenario that the online decision maker is with larger allowance and allocation rate at each round; see more details in~\ref{sec:Step-I}. 

Another related problem is online knapsack packing problem~\cite{zhou2008budget}. In the problem, items are associated with weight and value and come online. An online decision maker with a capacity limited knapsack determines whether to pack the item at its arrival to maximize the total value while guaranteeing that the total weight would not exceed the capacity. A single knapsack problem with infinitesimal assumption is studied in~\cite{zhou2008budget} with the application on key-work bidding.  It can be viewed as a special case of the one-way trading problem~\cite{Sun2020Competitive}, which is covered by the online optimization problem with a single inventory~\cite{cr_pursuit2018}. The fractional multiple knapsack packing problem with unit weight is studied in~\cite{Sun2020Competitive}, where the decision maker can pack any fraction of the item instead of $0/1$ decision. Our problem is also related to the online packing problem~\cite{buchbinder2009design,azar2016online} where the authors consider general packing constraints. Here, we focus on specific inventory constraints and allocation constraints. 


\section{Problem Formulation}\label{sec:problem}

In the section, we formulate the optimal allocation problem with multiple inventories. We discuss the practical online scenario and the performance metric for online algorithms. We further discuss the state-of-the-art to the problem. We summarize the important notions in Table~\ref{tab:notation}.


\begin{table}[t]
\centering
\caption{Notation Table.}
\label{tab:notation}
\begin{tabular}{|l|l|}
\hline
\textbf{Notation} & \textbf{Meaning}  \\ \hline \hline
$g_{i,t}(\cdot)$ &Revenue function of inventory $i$ at slot $t$ \\
$v_{i,t}$ & Allocation of inventory $i$ at slot $t$ \\
$C_i$ & Capacity of inventory $i$ \\ 
$A_t$ & Allowance of total allocation among all inventories at slot $t$\\  
$\delta_{i,t}$ & The maximum allocation of inventory $i$ at slot $t$\\  
$\theta$ & \specialcell{$\theta=p_{\max}/p_{\min}$, where \\$[p_{\min},p_{\max}]$ is the  range of the gradient of revenue functions}\\ \hline

$\hat{a}_{i,t}$ & Online allowance allocation to inventory $i$ at slot $t$\\  
$\hat{v}_{i,t}$ & Online inventory allocation of inventory $i$ at slot $t$\\ 
 
$\eta_{i,t}$ & Online revenue of inventory $i$ up to slot $t$\\
$\tilde{OPT}_{i,t}$ & \specialcell{Optimal objective of \textsf{\probname$_i$} given allowance allocations \\ $\;\;$ and revenue functions up to slot $t$}\\ 
$OPT_t$ & Optimal offline total revenue of \probname~up to slot $t$\\


$\Phi(\pi)$ & Maximum total online allocation of \textsf{CR-Pursuit($\pi$)}
 \\\hline
\end{tabular}
\end{table}

\subsection{Problem Formulation}

We consider $N$ inventories, and a decision period with $T$ slots. We denote the capacity of inventory $i$ as $C_i$. At each slot $t\in[T]$, each inventory $i$ is associated with a revenue function $g_{i,t}(v_{i,t})$, which represents the revenue of allocating an amount of $v_{i,t}$ inventory $i$ at slot $t$. However, at each slot $t\in[T]$, we are restricted to allocate at most $\delta_{i,t}$ of inventory $i$, and the total allocation of all inventories at each slot $t$ is bounded by the allowance $A_t$. Our goal is to find an optimal allocation scheme that maximizes the total revenue in the decision period, while satisfying the allocation restrictions. 

Overall, we consider the following problem,
\begin{align}
  \probname:\;  \max \quad& \sum_{i\in[N]}{\sum_{t\in[T]}{g_{i,t}(v_{i,t})}}\\
    \text{s.t.}\quad & \sum_{t} v_{i,t}\leq C_i, \forall i\in[N], \label{eq:inventory}\\
         & \sum_{i} v_{i,t}\leq A_t,\forall t\in[T],\label{eq:allowance}\\
         & 0\leq v_{i,t}\leq \delta_{i,t}, \forall t\in[T],i\in[N],\label{eq:rate}
\end{align}
In \probname, we optimize the inventory allocation $\{v_{i,t}\}_{i\in[N],t\in[T]}$ to achieve the maximum total revenue subjecting to the capacity constraint of each inventory ~\eqref{eq:inventory}, the allowance constraint at each slot~\eqref{eq:allowance}, and the allocation rate limit constraint for each inventory at each slot~\eqref{eq:rate}.  Without loss of generality, we assume that $\delta_{i,t}\leq A_t, \forall i,t$. We consider the following set of revenue functions, denoted as $\mathcal{G}$,

\begin{itemize}
         \item $g_{i,t}(\cdot)$ is concave and differentiable with $g(0)=0$;
         \item $g_{i,t}'(v_{i,t})\in[p_{\min},p_{\max}], 
         \forall v_{i,t}\in[0,\delta_{i,t}]$.
\end{itemize}


We consider that $p_{\max}\geq p_{\min}>0$ and denote $\theta=p_{\max}/p_{\min}$. The revenue functions capture the case where the marginal revenue of allocating more inventory is non-increasing in the allocation amount but always between $p_{\min}$ and $p_{\max}$. We also discuss how we can apply our approach to different sets of revenue functions with corresponding applications in Sec~\ref{sec:discussion-general-application}. 



In the offline setting, the problem input is known in advance, and \probname~is a convex optimization problem with efficient optimal algorithms. 
However, in practice, we are facing an online setting as we describe next. 

\subsection{Online Scenario and Performance Metric}
\label{sec:online_setting}

In the online setting, we consider that the pre-known problem parameters include the class of revenue function $\mathcal{G}$ and corresponding range $[p_{\min},p_{\max}]$, the number of inventories $N$, and the capacity of each inventory $\{C_i\}_{i\in[N]}$.  Other problem parameters are revealed sequentially. More specifically, at each slot $t$, the online decision maker without the information of the decision period $T$ is fed the revenue functions $\{g_{i,v}(\cdot)\}_{i\in[N]}$, the allowance $A_t$ and the allocation limits $ \{\delta_{i,t}\}_{i\in[N]}$. We needs to irrevocably determine the allocation at slot $t$, i.e., $ \{v_{i,t}\}_{i\in[N]}$. After that, if the decision period ends, we stop and know the information of $T$. Otherwise, we move to next slot and continue the allocation. We denote a possible input as
 \begin{equation}
 \sigma=\left(T, \{g_{i,t}(\cdot)\}_{i\in[N],t\in[T]}, \{A_t\}_{t\in[T]}, \{\delta_{i,t}\}_{i\in[N],t\in[T]}\right)
 \end{equation}
We use the CR as a performance metric for online algorithms. The CR of an algorithm $\mathcal{A}$ is defined as,
\begin{equation}
    \mathcal{CR(A)} = \sup_{\sigma\in\Sigma} \frac{OPT(\sigma)}{ALG(\sigma)},  
\end{equation}
where $\sigma$ denotes an input, $OPT(\sigma)$ and $ALG(\sigma)$ denote the  offline optimal objective and the online objective applying $\mathcal{A}$ under input $\sigma$, respectively. We use $\Sigma$ to represent all possible input we are interested in.  Specifically, 
\begin{equation}
    \Sigma\triangleq \{\sigma\left|T\in \mathbb{Z}^+,g_{i,t}(\cdot)\in\mathcal{G}, A_t\geq 0, \delta_{i,t}\geq 0,\forall i\in[N],t\in[T]\right.\},
\end{equation}
In competitive analysis, we focus on the worst-case guarantee of an online algorithm, which is defined by the maximum performance ratio between the offline optimal and the online objective of the algorithm.  In the online setting, we are facing two main challenges, 1) the decision maker does not known the future revenue functions, while the allocation now would affect the future decision due to the capacity constraint~\cite{cr_pursuit2018}; and 2) the online allowance constraints and rate constrains couple the decisions across the inventories, which are with known challenges in online matching and allocation problems~\cite{KALYANASUNDARAM20003bmatching,TCS-allocation}. 

\subsection{State of the Art}

The online problem has been studied in ~\cite{Sun2020Competitive} under the same revenue function set $\mathcal{G}$. A discrete counterpart of the problem is studied in~\cite{ma2020algorithms}. In some special cases, our revenue functions cover the linear functions with slopes between $[p_{\min},p_{\max}]$. In addition, when $p_{\max} = p_{\min}$, our problem reduces to maximizing the total amount of allocation, which has been widely studied in~\cite{azar2006maximizing,Deng2018Capacity,karp1990optimal,Devanur2013Randomized,KALYANASUNDARAM20003bmatching}. Here, we introduce a novel divide-and-conquer approach for the problem and show the our approach can achieve a close-to-optimal CR. In Sec.~\ref{sec:discussion-general-application}, we also show that our approach can be applied to different sets of revenue functions, which have not been studied in the existing literature.

Before proceeding, we discuss the two most relevant works, namely~\cite{ma2020algorithms} and ~\cite{Sun2020Competitive} in the literature. They both apply the online primal-dual analysis and design threshold functions for the online decision-makings of \probname. While the work~\cite{ma2020algorithms} studied a discrete setting differing from the continuous setting studied in~\cite{Sun2020Competitive} and our work, it is known in~\cite{niazadeh2013unified} that the same threshold-based function can be directly applied to the continuous setting, attaining the same CR. In the following, let us reproduce the algorithm for the continuous setting and the CR achieved by~\cite{ma2020algorithms}. The CR is better than the one proposed in~\cite{Sun2020Competitive}, and thus we deem it as the state of the art in the literature. We also compare the CR they achieve and ours in Sec.~\ref{sec:online}; see an illustration example in Fig.~\ref{fig:cr}.

Let~$\phi_i(w)$ denote the threshold function for each inventory~$i$, where~$w$ refers to the amount of allocated capacity of the inventory and~$\phi_i(w)$ can be viewed as a pseudo-cost of the allocation. At each slot, the algorithm determines the allocated amount~$v_{i,t}$ of inventory~$i$ at slot~$t$ by maximizing the per-round pseudo-revenue, which is the difference between the revenue and the threshold function, i.e.,


\begin{align}
    \textsf{(P\&D):}\; \max \quad& \sum_{i\in[N]}\left({g_{i,t}(v_{i,t})}-\int_{w_{i,t-1}}^{w_{i,t-1}+v_{i,t}} \phi_i(w) dw\right) \label{eq:primal-dual-algorithm}\\
    \text{s.t.}\quad &          \sum_{i} v_{i,t}\leq A_t,\label{eq:allowance-D}\\
         & 0\leq v_{i,t}\leq \delta_{i,t}, \forall i\in[N],
\end{align}
where $w_{i,t-1}$ is the total online allocation of inventory $i$ from the first slot to slot $t-1$. The algorithm is proposed in~\cite{Sun2020Competitive}, which can be viewed as a continuous reinterpretation of the discrete algorithm in~\cite{ma2020algorithms}.  According to Appendix E of~\cite{ma2020algorithms}, we can apply the following threshold function given that 
the gradient of $g_{i,t}(v_{i,t})$ is uniformly bounded in range $[p_{\min},p_{\max}]$,

\begin{equation}\label{eq:thresholdfunction}
    \phi_i(w) = \begin{cases}
    (p_{\min})^{\frac{1-w/C_i}{1-\chi}}(p_{\max})^{\frac{w/C_i}{1-\chi}}, & w\in[0,\chi\cdot C_i]; \\
    p_{\min}\cdot\frac{e^{w/C_i}-1}{e^{\chi}-1},&  w\in[\chi\cdot C_i,C_i];
    \end{cases}
\end{equation}
where $\chi={W(\ln(\theta)\cdot e^{\ln(\theta)-1})-\ln\theta+1} $ ($W(\cdot)$ is the Lambert-W function). 

\begin{prop}[\cite{ma2020algorithms}]
\label{thm:state-of-the-art}
    With the threshold function~(\ref{eq:thresholdfunction}), the threshold-based algorithm can achieve a competitive ratio of 
\begin{equation}\label{eq:cr_state_of_art}
    \tilde{\chi} = \frac{1}{1-e^{-\chi}}.
\end{equation}
\end{prop}

We provide the proof in Appendix~\ref{app:state-of-the-art}. We note that the proofs in~\cite{Sun2020Competitive} and~\cite{ma2020algorithms} follow similar ideas based on the online primal-dual framework but are different in presentations 
as one is discussing in the continuous setting~\cite{Sun2020Competitive} and the other in the discrete setting~\cite{ma2020algorithms}. As we are considering the continuous setting, our proof follows the same presentation discussed in~\cite{Sun2020Competitive} and applies the properties of the threshold function~\eqref{eq:thresholdfunction} discussed in~\cite{ma2020algorithms}. 

\section{CR-Pursuit for Single Inventory Problem}
\label{sec:single}
In this section, we first discuss the problem with single inventory. We extend the \textsf{CR-Pursuit} in~\cite{cr_pursuit2018} to cover the rate limit constraint and provide additional insights that will facilitate our algorithm design under the multiple inventories case. 

In the single inventory case, \probname~is reduced to the following problem,
\begin{align}
     \textsf{{\probnameS}}: \; \max & \sum_{t} g_{t}(v_{t})\\
    \text{s.t.} \; & \sum_{t} v_{t}\leq C\\
    \text{var.} \; & 0 \leq v_{t}\leq \delta_{t}, \forall t,
\end{align}
where $C$ denotes the capacity of the inventory,  $g_t(v_t)$ represents the revenue of allocating $v_t$ quantity of inventory at slot $t$, and $\delta_t$ is the rate limit restricting the maximum allocation at slot $t$. The goal is still to maximize the total revenue by determining the inventory allocation $v_t$ at each slot. We focus on the online setting described in Sec.~\ref{sec:online_setting} with $N$ specified to be one. We note that the \probnameS~has been studied in~\cite{Sun2020Competitive}. Also, the case under different assumptions on revenue functions and without rate limit has also been studied in~\cite{cr_pursuit2018}.    

Here, we generalize the results in~\cite{cr_pursuit2018} to consider revenue function set $\mathcal{G}$ and  involve rate limit constraint. 
 The online algorithm \crp, proposed in~\cite{cr_pursuit2018}, is a single-parametric online algorithm with $\pi$ as the parameter. At slot $t$, the algorithm first computes the optimal value of \probnameS~given the input revenue functions and rate limits up to $t$, which we denote as $OPT_S(t)$. It then determines the allocation $\hat{v}_t$ at slot $t$ such that
\begin{equation}
    \label{eq:cr-pursuit-single-online-decision}
    g_t(\hat{v}_t) = \frac{1}{\pi} \left(OPT_S(t)-OPT_S(t-1)\right).
\end{equation}

Under \crp, we define the maximum total allocation of the algorithm,
\begin{equation}
\label{eq:maximum_allocation_cr_pursuit}
    \Phi(\pi) \triangleq \sup_{\sigma\in\Sigma}\sum_{t}{\hat{v}_t},
\end{equation}
where $\hat{v}_t$ is determined by~\eqref{eq:cr-pursuit-single-online-decision}. By design, we clearly have the following properties of the 
CR-Pursuit($\pi$) algorithm.

\begin{lem}
\label{thm:single-rate-limit-feasibility}
We have $\hat{v}_t \leq \frac{1}{\pi}\cdot\delta_t$.
\end{lem}

\begin{proof}
As $g_t(v_t)$ is increasing and concave function, and 

\begin{equation}
   g_t(\hat{v}_t) = \frac{1}{\pi} \left(OPT_S(t)-OPT_S(t-1)\right)\leq \frac{1}{\pi}\cdot g_t(\delta_t)
\end{equation}
We have $\hat{v}_t \leq \frac{1}{\pi}\cdot\delta_t$.
\end{proof}
Lemma~\ref{thm:single-feasible-competitive} shows a upper bound on the online allocation, which guarantees the existence of $\hat{v}_t$ at each slot. It will also be useful for our algorithm design and CR analysis for \probname, which we will discuss in Sec.~\ref{sec:online}.

\begin{lem}
\label{thm:single-feasible-competitive}
\crp is feasible and $\pi$-competitive for \textsf{{\probnameS}} if $\Phi(\pi)\leq C$.
\end{lem}

\begin{proof}
Considering an arbitrary input $\sigma$, we first note that it is clear that $\hat{v}_t\leq \delta_t$ according to Lemma~\ref{thm:single-rate-limit-feasibility}. Then if $ \Phi(\pi)\leq C$, we have $\sum_{t}\hat{v}_t\leq C$, i.e., it satisfies the inventory constraint under input $\sigma$.

Summarizing~\eqref{eq:cr-pursuit-single-online-decision} over all $t$, we have, the online objective
\begin{equation}
    ALG(\sigma) = \sum_{t=1}^{T} g_t(\hat{v}_t)= \frac{1}{\pi} OPT_S(T)=  \frac{1}{\pi}\cdot OPT_S(\sigma),
\end{equation}
where $OPT_S(\sigma)$ is the optimal offline objective. Thus the algorithm is $\pi$-competitive. 

\end{proof}


Lemma~\ref{thm:single-feasible-competitive} shows that we can rely on characterizing $\Phi(\pi)$ to optimize the choice of $\pi$ in \crp. Further, we can interpret $\Phi(\pi)$ as the inventory the online algorithm $\textsf{CR-Pursuit($\pi$)}$ needs to maintain $\pi$-competitive. For example, suppose for an online algorithm, we now can utilize $\Phi(1)$ capacity of the inventory while the capacity of the offline optimal remains $C$. Then we can run $\textsf{CR-Pursuit(1)}$ and achieve that same performance as the offline optimal, i.e., 1-competitive.


We have the following results on the upper bound on $\Phi(\pi)$.

\begin{lem}
\label{thm:upper-bound-class-I}
We have 
\begin{equation}
    \Phi(\pi)\leq \frac{\ln\theta+1}{\pi}\cdot C
\end{equation}
\end{lem}



We summarize the proof idea of Lemma~\ref{thm:upper-bound-class-I} here while leaving the detailed proof to Appendix~\ref{proof:upper-bound}. We first notice that a more general results in~\cite{cr_pursuit2018} can be extended to the case with rate limit constraint, as shown in Proposition~\ref{thm:upper-bound-general} in Appendix~\ref{proof:upper-bound}. Although the results in~\cite{cr_pursuit2018} do not cover the revenue functions we consider here, it covers the revenue functions in the maximizer of $\Phi(\pi)$ as special cases. This observation leads to Lemma~\ref{thm:upper-bound-class-I}.

According to the above discussion, we can provide the competitive analysis of \crp~ for \probnameS. 
\begin{thm}
\label{thm:single-competitive-class-I}
 For \textsf{{\probnameS}}, 
 \textsf{CR-Pursuit($\pi_{1}$)} is $\pi_{1}$-competitive, where $\pi_{1}=\ln\theta+1$. And, it is optimal  among all online algorithms for the problem. 
\end{thm}
According to Lemma~\ref{thm:single-feasible-competitive} and Lemma~\ref{thm:upper-bound-class-I}, it is clear that \textsf{CR-Pursuit($\pi_{1}$)} is feasible and $\pi_{1}$-competitive.  Further, according to the results in~\cite{cr_pursuit2018,Sun2020Competitive}, we know that $\ln\theta+1$ is the lower bound or the optimal CR to \probnameS. 
 Thus, \textsf{CR-Pursuit($\pi_{1}$)} is also optimal.



\section{Online Algorithms for Multiple Inventory Problem}\label{sec:online}
In this section, we introduce our divide-and-conquer online algorithm \textsf{\algname($\pi$)} for \probname, where $\pi$ is a parameter to be specified. We first outline the algorithm structure. Following the structure, we then propose our general online algorithm for arbitrary $N$. We next show a simple and optimal online algorithm when $N$ is relatively small. Finally, we summarize our algorithm and provide the competitive analysis. An illustration of our approach and results is shown in Fig~\ref{fig:flowchart}.  

\subsection{Algorithm structure}\label{sec:alg-structure}
We consider a divide-and-conquer approach for deriving online algorithms for \probname. The general idea is that we can optimize \probname~ by first allocating the allowance at each slot to the inventories and then separately optimizing the allocation of each inventory given the allocated allowance. More specifically, we define the following subproblem for each $i\in[N]$,

\begin{align}
   \textsf{${\probname_i}$}: \; \max & \sum_{t} \tilde{g}_{i,t}(v_{i,t})\\
    s.t. & \sum_{t} v_{i,t}\leq C_i\\
         & 0 \leq v_{i,t}\leq a_{i,t}, \forall t,
\end{align}
where $a_{i,t}$ is the allocated allowance to user $i$ at slot $t$. $\tilde{g}_{i,t}(v_{i,t})$ is another algorithmic design space that allows us to exploit the online augmentation scenario when allocating the allowance allocation in Sec.~\ref{sec:Step-I}. Under the offline setting, we note that such a decomposition is of no optimality loss. For example, we can choose $\tilde{g}_{i,t}(v_{i,t})=g_{i,t}(v_{i,t})$  and set the allowance allocation  $a_{i,t}=v^*_{i,t}$ for all $i$ and $t$, where ${\{v^*_{i,t}\}}_{i\in[N], t\in[T]}$ is the offline optimal solution of \probname. Then, optimizing the subproblems separately given the allowances would reproduce the offline optimal solution. 

Following this structure, we can design an online algorithm that mainly consists of two steps at each slot $t$,
\begin{enumerate}
    \item \textsf{Step-I:} Determine the allowance allocation, $\{\hat{a}_{i,t}\}_{i\in[N]}$, irrevocably. 
    \item \textsf{Step-II:} Determine the the inventory allocation for each online \textsf{${\probname_i}$}, $\{\hat{v}_{i,t}\}_{i\in[N]}$, irrevocably.
\end{enumerate}
We note that this divide-and-conquer approach allows us to separately tackle the two main challenges of the problem. First, the revenue functions come online while the allocation across the decision period is coupled due to the capacity constraint for each inventory, which is mainly handled by \textsf{Step-II}. 
Second, the online allowance constraints and the rate constraints couple the decisions across the inventories, which we tackle in \textsf{Step-I}. We can directly apply $\{\hat{v}_{i,t}\}_{t\in[T]}$ as the output of an online algorithm for each inventory $i$. We can view the \textsf{Step-II} as solving \textsf{${\probname_i}$} in an online manner. 
More specifically, at each $t$, for each $i$, we observe input $\hat{a}_{i,t}$ (determined at \textsf{Step-I}) and $\tilde{g}_{i,t}(\cdot)$, and need to determine $v_{i,t}$ irrevocably.  In terms of feasibility, an immediate advantage is that, it satisfies the inventory constraint~\eqref{eq:inventory} if it is a feasible solution to \textsf{${\probname_i}$}. However, we need further care to make sure the satisfaction of the allowance constraint~\eqref{eq:allowance} and allocation rate limit~\eqref{eq:rate}. As for performance guarantees, we can first analyze the performance of the each step and then combine them to show the overall competitive analysis. In the following, we will discuss how the proposed online algorithm behaviors at both steps to ensure the feasibility and achieve a close-to-optimal CR.



\begin{figure}
    \centering
    \includegraphics[width=0.7\textwidth]{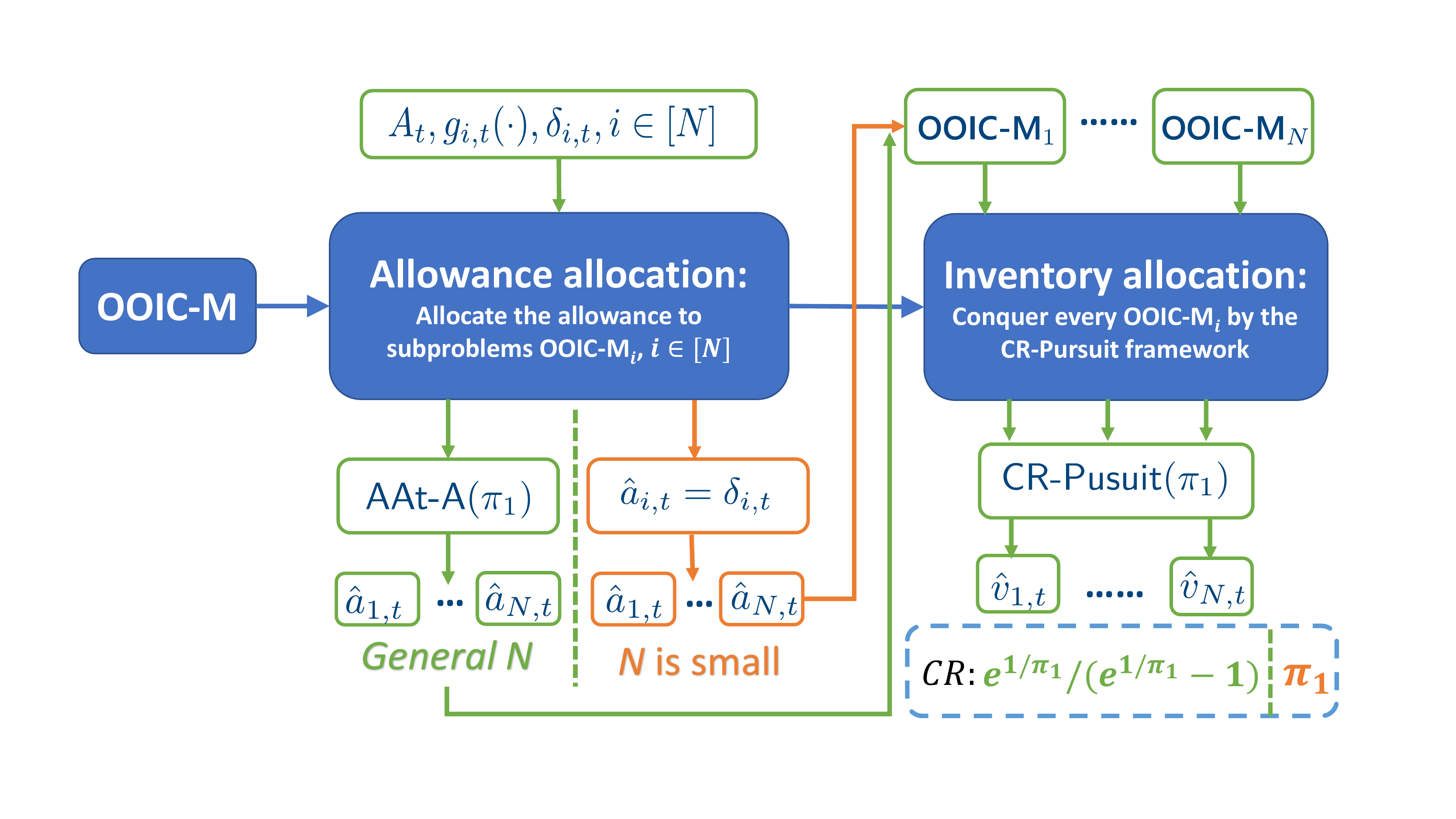}
    \caption{An Illustration of Our Divide-and-Conquer Approach and Results.}
    \label{fig:flowchart}
\end{figure}


\subsection{The \textsf{A\&P$_l$($\pi$)} Algorithm for General $N$}\label{sec:general-N}
In this subsection, we will propose an online algorithm for general number of inventories following the divide-and-conquer structure discussed in Sec.~\ref{sec:alg-structure}.  We denote the algorithm as \textsf{A\&P$_l$($\pi$)}, where $\pi$ is parameter to be specified. In the following, we first introduce the \textsf{Step-II} of \textsf{A\&P$_l$($\pi$)}, determining the allocation of \textsf{\probname$_i$} given the allowance from \textsf{Step-I}. We then introduce \textsf{Step-I}, determining the allowance of each inventories. We denote the allocated allowance from  \textsf{Step-I} as $\hat{a}_{i,t},\forall i,t$.

\subsubsection{\textsf{Step-II} of \textsf{A\&P$_l$($\pi$)}}\label{sec:Step-II}
In this step, \textsf{A\&P$_l$($\pi$)} determines the online inventory allocation for each \textsf{${\probname_i}$} given the allocated allowance (denoted as $\{\hat{a}_{i,t}\}_{i\in[N]}$) from \textsf{Step-I}. 



In \textsf{Step-II} of \textsf{A\&P$_l$($\pi$)}, it sets 
\begin{equation}
\label{eq:choose-of-g}
    \tilde{g}_{i,t}(v_{i,t}) = \pi\cdot g_{i,t}(\frac{v_{i,t}}{\pi}).
\end{equation}

We denote the optimal objective of \textsf{${\probname_i}$} given its input up to slot $t$ as $\tilde{OPT}_{i,t}$. We set $\tilde{OPT}_{i,0}=0$. At each slot $t$, it determines the allocation $\hat{v}_{i,t}$ such that it satisfies
\begin{equation}
    g_{i,t}(\hat{v}_{i,t}) = \frac{1}{\pi} \left(\tilde{OPT}_{i,t}-\tilde{OPT}_{i,t-1}\right).\label{eq:subproblem-crpursuit}
\end{equation}


While it looks similar to the \textsf{CR-Pursuit($\pi$)} algorithm discussed in  Sec.~\ref{sec:single} and~\cite{cr_pursuit2018}, we note that a major difference is that we are using the original revenue function $g_{i,t}(\cdot)$ to pursuit a fraction of $1/\pi$ of the optimal objective achieved over the revenue function $\tilde{g}_{i,t}$ instead of $g_{i,t}(\cdot)$. In general, $\tilde{g}_{i,t}$  defined in~\eqref{eq:choose-of-g} is no less than $g_{i,t}(\cdot)$ (take equality under the linear function case). Thus, it may be more difficult for the $\textsf{Step-II}$ to achieve the same performance ratio ($\pi_1$) 
between the optimal objective to the online objective for each \textsf{${\probname_i}$} as that in Sec.~\ref{sec:single}.  However, we would show the performance ratio remains achievable in Lemma~\ref{thm:step-II-class-I}. 
 The design of $\tilde{g}_{i,t}(\cdot)$ is important for achieving a better approximation ratio between the total optimal objective of the subproblems and the optimal objective of {\probname}, which we would discuss in \textsf{Step-I}, Sec.~\ref{sec:Step-I}. 

To analyze \textsf{Step-II}, we first propose the following proposition on the properties of the online allocation $\hat{v}_{i,t}, \forall i,t$.
\begin{lem}
\label{thm:step-II-online-bound}
We have $\hat{v}_{i,t}\leq \frac{1}{\pi} \cdot \hat{a}_{i,t},\forall i,t$.
\end{lem}
\begin{proof}
We have 
\begin{equation}
    g_{i,t}(\hat{v}_{i,t}) = \frac{1}{\pi} \left(\tilde{OPT}_{i,t}-\tilde{OPT}_{i,t-1}\right)\leq\frac{1}{\pi}\cdot \pi\cdot g_{i,t}(\frac{\hat{a}_{i,t}}{\pi})\leq g_{i,t}(\frac{\hat{a}_{i,t}}{\pi})
\end{equation}
Thus, as $g_{i,t}(\cdot)$ is increasing, we conclude $\hat{v}_{i,t}\leq \frac{1}{\pi} \cdot \hat{a}_{i,t}$.
\end{proof}
While this is a simple observation on the online solution, it plays an important role on designing the allowance allocation in \textsf{Step-I}, Sec.~\ref{sec:Step-I}  and improving the overall performance of our online algorithm \textsf{\algname$_l$($\pi$)} for \probname.
 
We denote the objective value of the online solution to \textsf{\probname$_i$} at slot $t$ as $\eta_{i,t}$, 
\begin{equation}
\label{eq:step2-online-ojective}
    \eta_{i,t}  \triangleq  \sum_{\tau=1}^{t} g_{i,\tau}(\hat{v}_{i,\tau})
\end{equation}

We provide performance analysis of \textsf{Step-II} in the following lemma. Recall that $\pi_1=\ln\theta+1$.

\begin{lem}
\label{thm:step-II-class-I}
 We have that for each $i\in[N]$, \textsf{Step-II} of \algname($\pi_{1}$) always produces a feasible solution to  \textsf{${\probname_i}$}, 
 and for any slot $t$,  
 the online objective
\begin{equation}
\label{eq:step-II-guarantee}
    \eta_{i,t}\geq \frac{1}{\pi_{1}}\cdot \tilde{OPT}_{i,t}.
\end{equation}
\end{lem}

\begin{proof}
The performance guarantee in~\eqref{eq:step-II-guarantee} is simply implied by~\eqref{eq:subproblem-crpursuit} when choosing $\pi=\pi_{1}$. 

We now show the feasibility. We note that, when choosing $\pi=\pi_{1}$,  \textsf{${\probname_i}$} with $\tilde{g}_{i,t}(\cdot)$ determined by~\eqref{eq:choose-of-g} and a factor of $\frac{1}{\pi_{1}}$ in objective value is equivalent to 
\begin{align}
   \textsf{R-${\probname_i}$}: \; \max & \sum_{t} g_{i,t}(z_{i,t})\\
    s.t. & \sum_{t} z_{i,t}\leq \frac{C_i}{\pi_{1}}\\
         & 0 \leq z_{i,t}\leq \frac{\hat{a}_{i,t}}{\pi_{1}}, \forall t
\end{align}
, where $z_{i,t} \triangleq v_{i,t}/\pi_1$. Then, determining the online allocation according to~\eqref{eq:subproblem-crpursuit} is equivalent to find $\hat{v}_{i,t}$ such that
\begin{equation}
\label{eq:step-II-equivalent-v}
    g_{i,t}(\hat{v}_{i,t}) = \hat{OPT}_{i,t}-\hat{OPT}_{i,t-1},
\end{equation}
where $\hat{OPT}_{i,t}$ is optimal objective of \textsf{R-${\probname_i}$} at slot $t$. It is clear that $\hat{v}_{i,t}\leq a_{i,t}/\pi_{1}$ as $\hat{OPT}_{i,t}-\hat{OPT}_{i,t-1}\leq g_{i,t}(\hat{a}_{i,t}/\pi_{1})$. Thus, the rate limit constraint in \textsf{${\probname_i}$} is satisfied.  

We note that the \textsf{R-${\probname_i}$} is a single inventory problem we discuss in Sec.~\ref{sec:single} with inventory capacity of $\frac{C_i}{\pi_1}$. The online decision we make according to~\eqref{eq:step-II-equivalent-v} suggests that we are running \textsf{CR-Pursuit($1$)} over online \textsf{R-${\probname_i}$}. According to Lemma~\ref{thm:upper-bound-class-I}, for \textsf{R-${\probname_i}$}, we have 
\begin{equation}
    \sum_{t}\hat{v}_{i,t}\leq \frac{\ln\theta+1}{1}\cdot \frac{C_i}{\pi_{1}}= C_i
\end{equation}
, noting that the inventory capacity of \textsf{R-${\probname_i}$} equals $C_i/\pi_{1}$. Thus, the online solution satisfies the capacity constraint in  \textsf{${\probname_i}$}.
\end{proof}

\subsubsection{\textsf{Step-I} of \textsf{A\&P$_l$($\pi$)}}\label{sec:Step-I}

The \textsf{Step-I} is to determine $\{\hat{a}_{i,t}\}_{i\in[N]}$, the allowance allocation to different inventories at each slot $t$. Our goal is to determine an allocation such that we can guarantee the a larger approximation ratio ($\triangleq\alpha$) between $\sum_{i\in[N]} \tilde{OPT}_{i,t}$ and $OPT_t$ at any slot $t$, i.e., $\sum_{i\in[N]} OPT_{i,t} \geq \alpha \cdot OPT_t$. Recall that $\tilde{OPT}_{i,t}$ is the optimal objective of \textsf{${\probname_i}$} up to slot $t$. And, $OPT_t$ is the optimal objective of \probname~up to slot $t$. 


As discussed in Sec.~\ref{sec:alg-structure}, we need further consideration to guarantee the satisfaction of the allowance constraint~\eqref{eq:allowance} and allocation rate limit~\eqref{eq:rate}. We characterize a sufficient condition on the allowance allocation such that the online solution $\{\hat{v}_{i,t}\}_{i\in[N]}$ determined at \textsf{Step-II} (as discussed in Sec.~\ref{sec:Step-II}) satisfies constraints~\eqref{eq:allowance} and ~\eqref{eq:rate}.

\begin{lem}
\label{thm:allowance-allocation-constraint}
If the allowance allocation at each slot $t$ satisfies
\begin{equation}
\label{eq:condition-for-feasibility}
   \sum_{i\in[N]} \hat{a}_{i,t}\leq \pi\cdot A_t, 0\leq \hat{a}_{i,t}\leq \pi \cdot \delta_{i,t},
\end{equation}
then the online solution $\{\hat{v}_{i,t}\}_{i\in[N]}$ determined by~\eqref{eq:subproblem-crpursuit} at \textsf{Step-II} satisfies the allowance constraint~\eqref{eq:allowance} and rate limit constraints~\eqref{eq:rate}.
\end{lem}

The idea is that according to Lemma~\ref{thm:step-II-online-bound}, $\hat{v}_{i,t}\leq \frac{1}{\pi}\hat{a}_{i,t}$. Together with~\eqref{eq:condition-for-feasibility}, it implies that we have $\sum_{i\in[N]} \hat{v}_{i,t}\leq  \frac{1}{\pi}\sum_i \hat{a}_{i,t}\leq A_t$ and $\hat{v}_{i,t}\leq\delta_{i,t}$. Lemma~\ref{thm:allowance-allocation-constraint} means that at each slot $t$, we can actually allocate $\pi$-time more total allowance to the subproblems while guaranteeing the online solution satisfies constraints~\eqref{eq:allowance} and ~\eqref{eq:rate}. We would show that it can help us significantly improve the approximation ratio $\alpha$ compared with allocating the allowance respecting constraints~\eqref{eq:allowance} and ~\eqref{eq:rate} directly. 

In the online literature, the \textsf{Step-I} problem is similar to the online ads allocation problem with free disposal studied in~\cite{Feldman2009Online} where each advertisers can be allocated more ads display slots than their pre-agreed numbers but would only count the most valuable ones. In \textsf{Step-I}, while the total allowance we allocate to inventory $i$ may exceed its capacity $C_i$, but the total revenue of inventory $i$, $\tilde{OPT}_{i,t}$, only counts the most valuable $C_i$ of them. Compared with~\cite{Feldman2009Online}, we further consider the setting that online decision maker can allocate $\pi$-time more allowance with a $\pi$-time relaxer rate limit constraint at each slot. We call it \emph{allowance augmentation scenario}.  We note that this is different with other works in online literature, e.g.,~\cite{azar2006maximizing,Deng2018Capacity,Deng2019Capacity}, where the authors consider the capacity augmentation scenario, i.e., how one can improve the performance guarantee (in particular, CR) of online algorithms when the online decision maker is equipped with more inventory capacity compared with the offline optimal. 
Here, we are the first one to consider allowance augmentation scenario. 

At \textsf{Step-I} of \textsf{A\&P$_l$($\pi$)}, we determine the allowance allocation by solving the following problem at each $t$. We call the problem  \textsf{\allallt-A($\pi$)} standing for Allowance Allocation at slot $t$ with Augmentation.

\begin{align}
   \textsf{\allallt-A($\pi$)}: \quad {\max}\quad & \sum_i\left(\tilde{g}_{i,t}( \hat{a}_{i,t})-\int_0^{\hat{a}_{i,t}}\Psi_{i,t}(a)\; da\right) \label{eq:pseudo-revenue-function}\\
    \text{s.t.} \quad & \sum_{i} \hat{a}_{i,t} \leq \pi\cdot A_t \label{eq:AAt-allowance-C} \\
    &  0\leq \hat{a}_{i,t} \leq \pi\cdot \delta_{i,t},\forall  i\in[N] \label{eq:AAt-rate-C} 
\end{align}
In  \textsf{\allallt-A($\pi$)}, $\{\hat{a}_{i,t}\}_{i\in[N]}$ is the allowance allocation at slot $t$. $\tilde{g}_{i,t}(\cdot)$ is defined as in~\eqref{eq:choose-of-g}. $\Psi_{i,t}(a)$  is define as follows.
\begin{equation}\label{eq:defn_Phi-C} 
\Psi_{i,t}(a) = f_i(C_i)\cdot G_{i,t}(C_i,a)-\frac{1}{\pi\cdot C_i}\int_0^{C_i} G_{i,t}\left(x,a\right)\cdot f_i(x)\; dx.  
\end{equation}
where  $f_i(x)$ and $G_i(x,a)$ are defined as,
\begin{equation}
\label{eq:defn_f-C}
    f_i(x)=\frac{1}{\pi\cdot C_i}\frac{1}{e^{\frac{1}{\pi}}-1}\cdot e^{x/{\left(\pi\cdot C_i\right)}},
\end{equation}

\begin{align}
G_{i,t}\left(x,a\right) = \max\quad & \sum_{\tau\in[t]}\tilde{g}_{i,\tau}( v_{i,\tau}) \label{eq:defn_G-C}\\
\text{s.t.}\quad  & \sum_{\tau\in[t]} v_{i,\tau}\leq x\label{eq:dfn_G_x}\\
& 0\leq v_{i,t}\leq a\\
 &0\leq v_{i,\tau}\leq \hat{a}_{i,\tau}, \forall \tau\in [t-1].
\end{align}

We can show that  \textsf{\allallt-A($\pi$)} is a convex optimization problem with simple linear packing constraints by checking that $\Psi_{i,t}(a)$ is non-decreasing in $a$ (shown in Proposition~\ref{thm:psi_increment} in Appendix~\ref{app:thm:step-I}). We can solve it using projected gradient descent where at each step an evaluation of $\Psi_{i,t}(a)$ is required. Although we do not have a close form of $G_{i,t}(\cdot)$, we can evaluate $\Psi_{i,t}(a)$ efficiently using numerical integration methods. 

Our algorithm in \textsf{Step-I}, i.e., solving \textsf{\allallt-A($\pi$)}, can be view as a continuous counterpart of the exponential weighting approach proposed in~\cite{Feldman2009Online}. $f_{i,t}(\cdot)$ defines the weight on per-unit revenue of $G_{i,t}(x,a)$, the optimal allocation of inventory $i$. $\Phi_{i,t}(a)$ is the exponential weighting total revenue of the optimal allocation of inventory $i$, which matches the dynamic threshold defined in~\cite{Feldman2009Online}. 
Here, we provide two novel understandings for exploiting the augmentation scenario. First, we redesign the weighted function $f_{i,t}(\cdot)$ which tunes the weight according to the allowance augmentation level $\pi$. Second, we design the revenue function $\tilde{g}_{i,t}(v_{i,t})=\pi\cdot g_{i,t}(v_{i,t}/\pi)$ to ensure that increasing the allowance allocation for inventory $i$ can substantially increase the revenue when compared with $g_{i,t}(v_{i,t})$ in the offline problem. If we simply apply $g_{i,t}(v_{i,t})$ directly, we would suffer the diminishing return of concave function and can not fully exploit the augmentation.  

We can show the approximation guarantee of \textsf{\allallt-A($\pi$)} in the following theorem.

\begin{thm}
 \label{thm:step-I}
 Given the allowance allocation $\hat{a}_{i,t}$ by solving \textsf{\allallt-A($\pi$)}, we have
 \begin{equation}
     \sum_{i\in[N]} \tilde{OPT}_{i,t} \geq \alpha(\pi) \cdot OPT_t,
 \end{equation}
 where $\alpha(\pi) = \pi\frac{e^\frac{1}{\pi}-1}{e^\frac{1}{\pi}}$. Furthermore, $\alpha(\pi)$ equals $\frac{e-1}{e}$ when $\pi=1$ and $1$ when $\pi\to\infty$.
\end{thm}

The proof of Theorem~\ref{thm:step-I} is provided in Appendix~\ref{app:thm:step-I}.  Our proof follows the online primal-and-dual analysis in~\cite{Feldman2009Online,buchbinder2009design}. We use the dual problem of \probname~as a baseline for comparison. By carefully design or update the dual variable at each slot, we show that our increment in the total optimal objective of all subproblems \textsf{\probname$_i$} is at least a fraction of $\alpha(\pi)$ of the increment on the objective of dual \probname~at  each slot $t$. This directly leads to Theorem~\ref{thm:step-I}.

\textbf{Remark.} The results we show in Theorem~\ref{thm:step-I} is with broader application scenarios and of independent interest. When all the revenue function is linear with a constant slope, i.e., all inventories have the uniform unit price, the \textsf{Step-I} problem reduces to maximum the total amount of allocation, which is studied in~\cite{azar2006maximizing,KALYANASUNDARAM20003bmatching}. Our result (Theorem~\ref{thm:step-I}) implies that, when there is no allowance augmentation (i.e., $\pi=1$), it reproduces the CR $\frac{e}{e-1}$ as shown in~\cite{azar2006maximizing,KALYANASUNDARAM20003bmatching}. Also, when considering linear revenue functions, our problem  can be viewed as a continuous counterpart of the online ad allocation problem with free disposal studied in ~\cite{Feldman2009Online}. We recover the $\frac{e}{e-1}$ CR when there is no allowance augmentation (i.e., $\pi=1$). In both case, our results generalize to the case with allowance augmentation and show an improved CR of $\frac{1}{\pi}\frac{e^\frac{1}{\pi}}{e^\frac{1}{\pi}-1}$ with $\pi$-time augmentation, which tends to one when $\pi\to\infty$. 

We also note that Theorem~\ref{thm:step-I} holds for arbitrary increasing and differential concave functions starting from the origin, not restricted to the revenue functions we consider in set $\mathcal{G}$. This would be useful for generalizing our approaches to a broader application area with different sets of revenue functions beyond $\mathcal{G}$, which we will discuss in Sec.~\ref{sec:discussion-general-application}.



\subsubsection{Competitive analysis of \textsf{A\&P$_l$($\pi$)}}
We first summarize \textsf{A\&P$_l$($\pi$)}. At each slot $t$, (\textsf{Step-I}) it solves  \textsf{\allallt-A($\pi$)} to obtain the allowance allocation $\hat{a}_{i,t}, \forall i\in[N]$, and (\textsf{Step-{II}}) it determines  $\hat{v}_{i,t}$ according to~\eqref{eq:subproblem-crpursuit}, for all $i\in[N]$.

We then show its performance guarantee off \probname~in the following theorem. 

\begin{thm}
\label{thm:AP-general-N}
The \textsf{A\&P$_l$($\pi_{1}$)} algorithm is $\frac{e^\frac{1}{\pi_1}}{e^\frac{1}{\pi_1}-1}$-competitive for \probname. 
\end{thm}

\begin{proof}

We fist show the feasibility of \textsf{A\&P$_l$($\pi_{1}$)}. By solving \textsf{\allallt-A($\pi_1$)}, $\{\hat{a}_{i,t}\}_{i\in[N],t\in[T]}$ satisfies condition~\eqref{eq:condition-for-feasibility} with $\pi=\pi_1$ in Lemma~\ref{thm:allowance-allocation-constraint}. According to  Lemma~\ref{thm:allowance-allocation-constraint}, the online solution satisfies the allowance constraint and rate limit constraint of \probname. Beside, according to Lemma~\ref{thm:step-II-class-I}, the online solution  is always feasible to \textsf{\probname$_i$}, i.e., it satisfies the capacity constraint of \probname. We conclude  the online solution of\textsf{A\&P$_l$($\pi_{1}$)} is feasible.

As for the CR, combining Theorem~\ref{thm:step-I} we obtain in \textsf{Step-I} of \textsf{A\&P$_l$($\pi_{1}$)} and Lemma~\ref{thm:step-II-class-I} in \textsf{Step-II}, we have that the online objective  of  \textsf{\algname$_l$($\pi_{1}$)},
\begin{equation}
     \sum_{i\in[N]}\eta_{i,t}\geq  \frac{1}{\pi_{1}}\cdot \sum_{i\in[N]} \tilde{OPT}_{i,t} \geq  \frac{1}{\pi_{1}}\alpha(\pi_{1}) \cdot OPT_t=\frac{e^\frac{1}{\pi_1}-1}{e^\frac{1}{\pi_1}} \cdot OPT_t, \forall t.
 \end{equation}
Thus, at the final slot $T$, we also have $\sum_{i\in[N]}\eta_{i,T}\geq \frac{e^\frac{1}{\pi_1}}{e^\frac{1}{\pi_1}-1} \cdot OPT_T$, and we conclude that \textsf{A\&P$_l$($\pi_{1}$)}  is $\frac{e^\frac{1}{\pi_1}}{e^\frac{1}{\pi_1}-1}$-competitive.

\end{proof}

We note that $\frac{e^\frac{1}{\pi}}{e^\frac{1}{\pi}-1}\leq \pi+1$. Thus, comparing with the result under the single-inventory case shown in Theorem~\ref{thm:single-competitive-class-I}, 
 Theorem~\ref{thm:AP-general-N} implies that we can achieve a CR with at most an additive constant (one) for the case with arbitrary number of inventories. Also, $\frac{e^\frac{1}{\pi}}{e^\frac{1}{\pi}-1}\sim\pi$, i.e., $\lim_{\pi\to\infty}\frac{e^\frac{1}{\pi}}{e^\frac{1}{\pi}-1}/\pi=1$. It shows that the CR we achieve for \probname~under arbitrary number of inventories is asymptotically equivalent to the one under the single-inventory case.

\subsection{A Simple Algorithm for Small $N$}\label{sec:small-N}
From the design of our divide-and-conquer approach, we note that our online algorithm is able to allocate $\pi$-times more allowance to the subproblems according to Lemma~\ref{thm:allowance-allocation-constraint}. It reveals that when the number of the inventory is small (e.g., less than $\pi$), the allowance constraint could becomes redundancy in our design. Leveraging the above insight, we show a simple and optimal online algorithm for \probname~when $N$ is relative small compared with $\theta$. More specifically, 
 we consider the case that $N\leq\pi_{1}$. 
 We denote our online algorithm as \textsf{A\&P$_s$($\pi$)} with $\pi$ as a parameter to be specified.  \textsf{A\&P$_s$($\pi$)} consists of two steps, where the first step is to allocate the allowance, and the second step is to pursuit a $\pi$ performance ratio for each subproblem. In the first step,  \textsf{A\&P$_s$($\pi$)} determines the allowance allocation as
\begin{equation}
\label{eq:small-N-allowance}
    \hat{a}_{i,t} = \delta_{i,t}.
\end{equation}
In the second step, for each \textsf{${\probname_i}$}, it chooses $\tilde{g}_{i,t}(v_{i,t})$ as $g_{i,t}(v_{i,t})$. We note that in such case \textsf{${\probname_i}$} reduces to the single inventory problem we discuss in Sec.~\ref{sec:single}. The \textsf{A\&P$_s$($\pi$)} determines the online solution running \crp. That is, it choose $\hat{v}_{i,t}$ such that

\begin{equation}
\label{eq:small-N-allocation}
    g_{i,t}(\hat{v}_{i,t}) =\frac{1}{\pi}(OPT_{i,t}-OPT_{i,t})
\end{equation}
where $OPT_{i,t}$ is the optimal objective of \textsf{${\probname_i}$} given $\{\hat{a}_{i,\tau}\}_{\tau\in[t]}$ and $\{\tilde{g}_{i,t}(\cdot)\}_{\tau\in[t]}$ at slot $t$.

\begin{thm}
\label{thm:AP-small-N}
The \textsf{A\&P$_s$($\pi_{1}$)} is $\pi_{1}$-competitive  when $N\leq\pi_{1}$. 
\end{thm}

\begin{proof}
We first check the feasibility of \textsf{A\&P($\pi_{1}$)}. The rate limit constraints and inventory constraints are directly guaranteed by the second step of \textsf{A\&P($\pi_{1}$)} where we run the \textsf{CR-Pursuit($\pi_{1}$)} (as shown in Theorem~\ref{thm:single-competitive-class-I}).  We then check the allowance constraints, for any $t$, we have
\begin{equation}
   \sum_{i} \hat{v}_{i,t} \leq  \sum_{i} \frac{1}{\pi_{1}}\cdot \hat{a}_{i,t} = \sum_{i} \frac{1}{\pi_{1}}\cdot \delta_{i,t}\leq \frac{1}{\pi_{1}} N\cdot A_t\leq A_t.
\end{equation}
Recall that we have $\hat{v}_{i,t} \leq \frac{1}{\pi_{1}}\cdot \hat{a}_{i,t}$ according to Lemma~\ref{thm:single-rate-limit-feasibility}, and without loss of generality, we consider $\delta_{i,t}\leq A_t$, as discussed in Sec.~\ref{sec:problem}.

We then show the performance analysis of the algorithm. It is clear that at each slot $t$, we have \begin{equation}
   \sum_i{OPT_{i,t}} \geq OPT_t,\forall t.
\end{equation}
where $OPT_t$ is the  optimal objective of \probname~at slot $t$. This is because $\sum_i{OPT_{i,t}}$ equals the optimal objective of of \probname~at slot $t$ without the allowance constraint. Then, we have the online objective
\begin{equation}
    \sum_{i,t} g_{i,t}(\hat{v}_{i,t}) = \frac{1}{\pi_{1}}\sum_i{OPT_{i,T}}\geq \frac{1}{\pi_{1}} OPT_T
\end{equation}
Thus, \textsf{A\&P($\pi_{1}$)} is $\pi_{1}$-competitive. 
\end{proof}

Theorem~\ref{thm:AP-small-N} shows that when the total number of inventories is relatively small compared with the uncertainty range of the revenue functions (i.e., $\theta$), we can reduce the multiple inventory problem to the single inventory case with the same performance guarantee.

\subsection{Summary of Our Proposed Online Algorithm}\label{sec:alg-cr}

\begin{algorithm}[!ht]
\caption{ \algname($\pi$) Algorithm \label{alg:online-algorithm}}
\begin{algorithmic}[1]
\STATE At slot $t$, $\{g_{i,t}(\cdot)\}_{i\in[N]}$, $A_t$, and $\{\delta_{i,t}\}_{i\in[N]}$ are revealed,
\IF{$N\leq\pi$}
\STATE Run \textsf{A\&P$_s$($\pi$)}, i.e., determine $\hat{a}_{i,t}=\delta_{i,t}$ as in~\eqref{eq:small-N-allowance}\\ and determine  $\hat{v}_{i,t}$ according to~\eqref{eq:small-N-allocation}, for all $i\in[N]$,
\RETURN $\{\hat{v}_{i,t}\}_{i\in[N]}$.

\ELSE 
\STATE Run \textsf{A\&P$_l$($\pi$)}:
\STATE \textsf{Step-I}: solve  \textsf{\allallt-A($\pi$)} to obtain the allowance allocation $\hat{a}_{i,t}, \forall i\in[N]$,
\STATE \textsf{Step-II}: determine  $\hat{v}_{i,t}$ according to~\eqref{eq:subproblem-crpursuit}, for all $i\in[N]$,
\RETURN $\{\hat{v}_{i,t}\}_{i\in[N]}$.
\ENDIF

\end{algorithmic}
\end{algorithm}

In this section, we summarize the our online algorithm, denoted as \algname($\pi$), and provide its performance analysis. An illustration of our approach and results is provided in Fig.~\ref{fig:flowchart}. We also compare our results with existing ones in Fig~\ref{fig:cr}. 

The pseudocode of \algname($\pi$)~ is provided in Algorithm~\ref{alg:online-algorithm}. Depending on the value of $N$ and $\theta$, we run either \textsf{A\&P$_s$($\pi$)} or \textsf{A\&P$_l$($\pi$)}. 
The CR of our online algorithm is shown in the following theorem. Recalled that $\pi_{1} = \ln\theta+1$. 

\begin{thm}
\label{thm:competitive-ratio-class-I}
Our online algorithm \algname($\pi_{1}$) achieves the following CR, 
\begin{equation}
\label{eq:compeptitve-ratio-class-I}
\mathcal{CR}_1(A\&P(\pi_{1})) = 
\begin{cases}
\pi_{1}, & \pi_{1} \geq N,\\
\frac{e^\frac{1}{\pi_1}}{e^\frac{1}{\pi_1}-1}, & \pi_{1}< N.
\end{cases}
\end{equation}

\end{thm}

Theorem~\ref{thm:competitive-ratio-class-I} simply combines the results we show in Theorem~\ref{thm:AP-small-N} and Theorem~\ref{thm:AP-general-N}. The CR we obtain is tight and optimal when $N$ is smaller than $\ln\theta+1$. This also recovers the results for single-inventory case discussed in Sec.~\ref{sec:single} and~\cite{Sun2020Competitive}. It is within an additive constant of one to the lower bound when $N$ is larger than $\ln\theta+1$. When $\theta=1$, our problem reduces to maximizing total allocation, and our result recovers the optimal CR $\frac{e}{e-1}$ achieved in~\cite{azar2006maximizing,KALYANASUNDARAM20003bmatching}. In~\cite{ma2020algorithms}, the authors show a CR that is within $[\pi_1,\frac{e^\frac{1}{\pi_1}}{e^\frac{1}{\pi_1}-1}]$, independent of $N$, and consistently lower than the one achieved in~\cite{Sun2020Competitive}. They also show that the CR they achieve is tight when $N$ tends to infinity.  While our achieved CR at large $N$ is worse compared with~\cite{ma2020algorithms}, the gap between is no greater than an additive constant of one. In additional, we achieve a better (and optimal) CR when $N$ is small. We provide an illustration on the CRs achieved by~\cite{ma2020algorithms,Sun2020Competitive} and ours in Fig.~\ref{fig:cr}. 

\begin{figure}[!h]
    \centering
    \includegraphics[width=0.75\textwidth]{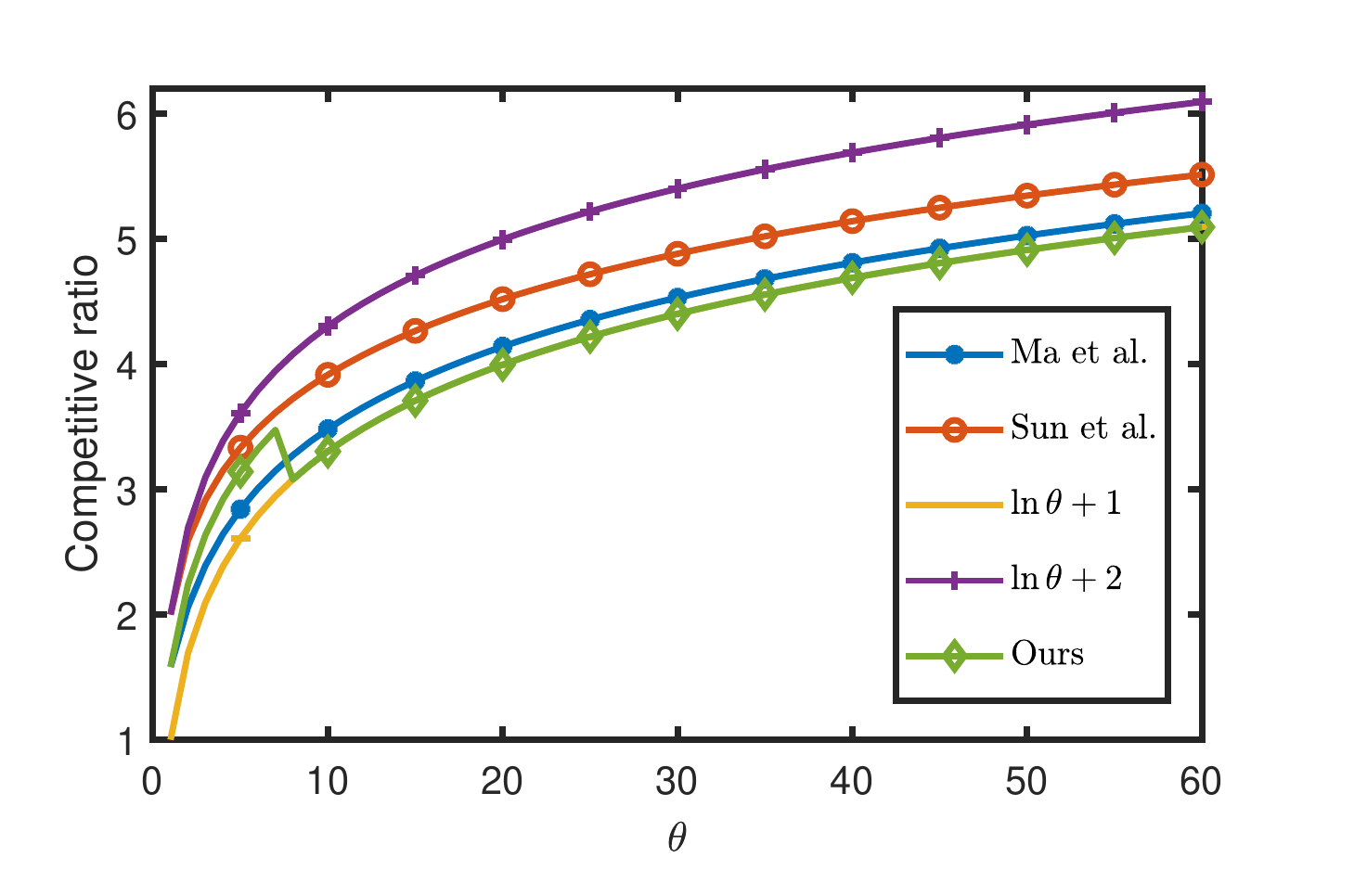}
    \caption{The competitive ratios obtained by Ma et al.~\cite{ma2020algorithms}, Sun et al.~\cite{Sun2020Competitive} and ours. We fix $N=3$ and vary $\theta$ from $1$ to $60$. 
    }
    \label{fig:cr}
\end{figure}





\section{Extension to General Concave Revenue Function}\label{sec:discussion-general-application}


In addition to the set of revenue functions $\mathcal{G}$ we discuss above, our divide-and-conquer approach can be applied under a broader range of functions with corresponding applications. For example, we widely observe the logarithmic functions (e.g., $\log(v+1)$) in wireless communication~\cite{Niv2012Dynamic,Niv2010Multi}, which is not covered by the revenue function set $\mathcal{G}$ when considering sufficiently large capacity. Also, the revenue functions in the application of one-way trading with price elasticity~\cite{cr_pursuit2018}.  In general, let us consider a given set of concave revenue functions with zero value at the origin; say $\tilde{\mathcal{G}}$. We define $\Phi_{\tilde{G}}(1)$ as the maximum online total allocation of running \textsf{CR-Pursuit(1)} under revenue functions $\tilde{\mathcal{G}}$ in the single inventory case (as defined in~\eqref{eq:maximum_allocation_cr_pursuit} with $\pi=1$). It represents the maximum capacity we required to  maintain the same performance of the offline optimal at all time. We have the following results for the \probname~under the set of revenue functions $\tilde{\mathcal{G}}$.

\begin{prop}
Suppose we can find $\tilde{\pi}$ such that for \probnameS, we have 
\begin{equation}
\label{eq:phi_1_condition}
\Phi_{\tilde{G}}(1)\leq \tilde{\pi} \cdot C.  
\end{equation}
We can run the \textsf{A\&P($\tilde{\pi}$)} for \probname~under $\tilde{\mathcal{G}}$. The competitive ratio of \textsf{A\&P($\tilde{\pi}$)} is given by,
\begin{equation}
\label{eq:compeptitve-ratio-class-general}
\mathcal{CR}_{\tilde{G}}(\textsf{A\&P(${\tilde{\pi}}$)} = 
\begin{cases}
\tilde{\pi}, & \tilde{\pi} \geq N,\\
\frac{e^\frac{1}{\tilde{\pi}}}{e^\frac{1}{\tilde{\pi}_1}-1}, & \tilde{\pi}< N.
\end{cases}
\end{equation}
\end{prop}

The proof follow the same idea as discussed in Sec.~\ref{sec:online} and is omitted here. When $\tilde{\pi}< N$, we simply recover Lemma~\ref{thm:step-II-class-I} with the condition~\eqref{eq:phi_1_condition}. Together with Theorem~\ref{thm:step-I}, we show the results in ~\eqref{eq:compeptitve-ratio-class-general} when $\tilde{\pi}\geq N$. As for the case $\tilde{\pi}< N$, we can recover Theorem~\ref{thm:AP-small-N} by noting that $\Phi_{\tilde{\mathcal{G}}}(\tilde{\pi})\leq\frac{1}{\tilde{\pi}}\Phi_{\tilde{\mathcal{G}}}(1)\leq C$ (due the concavity of the revenue functions), i.e., \textsf{CR-Pursuit($\tilde{\pi}$)} is feasible and $\tilde{\pi}$-competitive for \probnameS.

For example, we can consider the one-way trading with price elasticity problem with multiple inventories, where the single-inventory case is proposed in~\cite{cr_pursuit2018}. More specifically, we consider the following type of revenue function, which we denote as $\hat{\mathcal{G}}$. 
\begin{itemize}
      \item $g_{i,t}(v_{i,t})=\left(p_{i,t}-q_{i,t}(v_{i,t})\right)\cdot v_{i,t}$, $v_{i,t}\in [0,\delta_{i,t}]$, where $q_{i,t}(\cdot)$ is convex increasing function with $q_{i,t}(0)=0$ and $p_{i,t}\in[p_{\min},p_{\max}]$. \footnote{{There exist an $v$ (may be infinity) such the function $g_{i,t}(\cdot)$ is increasing in $[0,v]$ and decreasing in $[v,\infty]$. We only need to consider the case that $\delta_{i,t}\leq v$ as no reasonable algorithm would allocate more than $v$ at $g_{i,t}(\cdot)$. Thus, we consider that $g_{i,t}(\cdot)$ is increasing in $[0,\delta_{i,t}]$.}}
\end{itemize}

The revenue functions in $\hat{\mathcal{G}}$ consider that the price of selling (allocating) the inventory decreases as the supply (allocation) increases, which follows the basic law in supply and demand in microeconomics. In particular, the price elasticity is captured by a convex increasing function $q_{i,t}(\cdot)$, meaning that more supply would further decrease the price. The marginal revenue is bounded between $p_{\min}$ and $p_{\max}$ only at the origin and could be even zero otherwise, which implies $\hat{\mathcal{G}}$ is not covered by $\mathcal{G}$. 

According to Lemma 15 in~\cite{cr_pursuit2018} (while it does not consider rate limit constraint, we can check that the proof simply follows with limit constraint), we have 
\begin{equation}
    \Phi_{\hat{G}}(1)\leq 2\cdot (\ln\theta+1) \cdot C.  
\end{equation}

For \probname~under revenue  function set, $\hat{\mathcal{G}}$, we have 

\begin{prop}
\label{thm:competitive-ratio-class-II}
Let $\pi_2 =2\cdot(\ln\theta+1)$. \algname($\pi_{2}$) achieves the following competitive ratio for \probname~under revenue function set $\hat{\mathcal{G}}$, 
\begin{equation}
\label{eq:compeptitve-ratio-class-I}
\mathcal{CR}_{\hat{\mathcal{G}}}(A\&P(\pi_{2})) = 
\begin{cases}
\pi_{2}, & \pi_{2} \geq N,\\
\frac{e^\frac{1}{\pi_2}}{e^\frac{1}{\pi_2}-1}, & \pi_{2}< N.
\end{cases}
\end{equation}

\end{prop}

The CR we achieve is upper bounded by $2\ln\theta+3$, which is up to a constant factor multiplying the lower bound $\ln\theta+1$. This provides the first CR  of the \probname~under revenue function set $\hat{\mathcal{G}}$ with application to the one-way trading with price elasticity under multiple-inventory scenario. It is interesting to see whether we can fine a tiger bound on $\Phi_{\hat{\mathcal{G}}}(1)$ and achieve a better competitive ratio. In addition, while the determination of $\Phi_{\tilde{G}}(1)$ could be problem-specific, how we can show a general way for it would be another interesting future direction. 










\section{Concluding Remarks}\label{sec:conclusion}

In this work, we study the competitive online optimization problem under multiple inventories, \probname. {The online decision maker allocates the capacity-limited inventories to maximize the overall revenue, while the revenue functions and the allocation constraints at each slot come in an online manner.} Our key result is a divide-and-conquer approach that reduces the multiple-inventory problem to the single-inventory case with a small optimality loss in terms of the CR.  In particular, we show that our approach achieves the optimal when $N$ is small and is within an additive constant to the lower bound when $N$ is larger, when considering gradient bounded revenue functions. We also provide a general condition for applying our approach to broader applications with different interesting sets of revenue functions. In particular, for revenue functions appeared in one-way trading with price elasticity, our approach obtains an optimal CR for the problem that is up to a constant factor to the lower bound. As a by-product, we also provide the first allowance augmentation results for the online fractional matching problem and the online fraction allocation problem with free disposal. As a future direction, we are interested in how our divide-and-conquer approach can be used to solve other online optimization problems with multi-entity and how it can be applied in more application scenarios.



\bibliographystyle{unsrt}
\bibliography{ref.bib}

\appendix
\section*{APPENDIX}
\section{Proof of Proposition~\ref{thm:state-of-the-art}}\label{app:state-of-the-art}

Our proof follow the well-established online primal-and-dual approach~\cite{Sun2020Competitive,buchbinder2009design,ma2020algorithms,niazadeh2013unified}, etc.

We note that according to Appendix E of ~\cite{ma2020algorithms}. The threshold function is increasing and satisfies the following conditions. 
\begin{align}
    C_i \phi_i'(w)-\phi_i(w)\leq p_{\min}(\tilde{\chi}-1), & \;w\in(0,\chi\cdot C_i); \label{eq:condition-1}\\
    C_i \phi_i'(w)-\tilde{\chi}\cdot\phi_i(w)\leq 0, &\; w\in(\chi\cdot C_i,C_i). \label{eq:condition-2}
\end{align}

Accordingly, it simply applies that 
\begin{align}
    C_i \phi_i(w)-\int_{0}^{w}\phi_{i}(w)dw\leq p_{\min}(\tilde{\chi}-1)\cdot w, & \;w\in(0,\chi\cdot C_i); \label{eq:condition-1-re}\\
    C_i \phi_i(w)-\int_{0}^{w}\phi_i(w)\leq (\tilde{\chi}-1) \cdot \left(p_{min}\cdot \chi\cdot C_i+\int_{\chi\cdot C_i}^{w}\phi_i(w)dw\right), &\; w\in(\chi\cdot C_i,C_i). \label{eq:condition-2-re}
\end{align}

In the primal-and-dual framework, it applies the dual problem as a baseline for the offline optimal. The dual problem of \probname~at slot $T$,
\begin{align}
    \textsf{Dual-\probname}:\;\quad & \min \sum_{i,t\in[T]}h_{i,t}(\alpha_i+\beta_t) + \sum_{i}{C_i\alpha_i}+\sum_{t}{A_{t}\beta_{t}}\\
    \text{s.t.} \quad& \alpha_i\geq 0, \beta_t\geq 0, \forall t,i, \label{eq:dual-owfm-var-D}
\end{align}
where 
\begin{equation}
h_{i,t}(\lambda)=\max_{0\leq v\leq \delta_{i,t}} g_{i,t}(v)-\lambda\cdot v.
\label{eq:dual_h}
\end{equation}

We denote the online solution of the algorithm as $\bar{v}_{i,t}$, which is the optimal solution to problem  \textsf{(P\&D)} in~\eqref{eq:primal-dual-algorithm}. Recall that $w_{i,t}$ is the online total allocation of the algorithm from slot $1$ to slot $t$, i.e., 
\begin{equation}
    w_{i,t}=\sum_{\tau=1}^{t}\bar{v}_{i,\tau}
\end{equation} 
At each slot $t$, we denote the optimal dual solution of problem \textsf{(P\&D)} in~\eqref{eq:primal-dual-algorithm} associated with constraint~\eqref{eq:allowance-D} as $\hat{\beta}_t$. Note that according to KKT conditional, we have
\begin{equation}
\label{eq:dual-D-beta-allowance}
    \hat{\beta}_t\cdot(\sum_{i\in[N]}\bar{v}_{i,t}-A_t) = 0.
\end{equation}

We consider the following dual solution,
\begin{equation}
    \alpha_{i}  = \phi_{i}(w_{i,T}),\forall i; \beta_t  = \hat{\beta}_t,\forall t\in[T].
\label{eq:dual-baseline-owfm-D}
\end{equation}
We note that the dual variable satisfies the dual constraint~\eqref{eq:dual-owfm-var-D}. Then, we have
\begin{align}
    OPT_{T} \leq & \sum_{i\in[N],t\in[T]} h_{i,t}(\alpha_i+\beta_t) +\sum_{i\in[N]}C_i\alpha_i+\sum_{t\in[T]}A_t\beta_t\\
    = & \sum_{i\in[N],t\in[T]} h_{i,t}(\phi_{i}(w_{i,T})+\hat{\beta}_t) +\sum_{i\in[N]}C_i\phi_{i}(w_{i,T})+\sum_{t\in[T]}A_t\hat{\beta}_t\\
   \stackrel{(a)}{\leq} & \sum_{i\in[N],t\in[T]} h_{i,t}(\phi_{i}(w_{i,t})+\hat{\beta}_t) +\sum_{i\in[N]}C_i\phi_{i}(w_{i,T})+\sum_{t\in[T]}A_t\hat{\beta}_t \\
    \stackrel{(b)}{\leq} & \sum_{i\in[N],t\in[T]} \left[g_{i,t}(\bar{v}_{i,t})-(\phi_{i}(w_{i,t})+\hat{\beta}_t)\bar{v}_{i,t}\right]+\sum_{i\in[N]}C_i\phi_{i}(w_{i,T})+\sum_{t\in[T]}A_t\hat{\beta}_t\\
    \stackrel{(c)}{=}& \sum_{i\in[N],t\in[T]}\left[g_{i,t}(\bar{v}_{i,t})-\phi_{i}(w_{i,t})\bar{v}_{i,t}\right]+\sum_{i\in[N]}C_i\phi_{i}(w_{i,T})\\
    \stackrel{(d)}{\leq} &  \sum_{i\in[N],t\in[T]}g_{i,t}(\bar{v}_{i,t})+\sum_{i\in[N]}\left[C_i\phi_{i}(w_{i,T})-\int_0^{w_{i,T}} \phi_i(w) dw\right]\\
    \stackrel{(e)}{\leq} & \tilde{\chi}\cdot\sum_{i\in[N],t\in[T]}g_{i,t}(\bar{v}_{i,t})
\end{align}
where (a) is due to the non-decreasing of $\phi_i(\cdot)$ and $h_{i,t}(\cdot)$ defined in~\eqref{eq:dual_h}; (b) is due to $\bar{v}_{i,t}$ is the optimal solution to~\eqref{eq:dual_h} when $\lambda=\phi_{i}(w_{i,t})+\hat{\beta}_t$ by checking that the KKT conditions of the problem \textsf{(P\&D)} in~\eqref{eq:primal-dual-algorithm}; (c) is according to~\eqref{eq:dual-D-beta-allowance}; (d) is according to $\phi_{i}(w_{i,t})\bar{v}_{i,t}\geq \int_{w_{i,t}-\bar{v}_{i,t}}^{w_{i,t}} \phi_i(w) dw$; (e) is by the fact that
\begin{equation}
\label{eq:key-D}
   C_i\phi_{i}(w_{i,T})-\int_0^{w_{i,T}} \phi_i(w) dw \leq (\tilde{\chi}-1)\sum_{t\in[T]}g_{i,t}(\bar{v}_{i,t}).
\end{equation}

We show~\eqref{eq:key-D} in the following. When $w_{i,T}\leq \chi\cdot C_i$, it directly follows~\eqref{eq:condition-1-re} and $\sum_{t\in[T]}g_{i,t}(\bar{v}_{i,t})\geq p_{min}\cdot w_{i,T}$.  When $w_{i,T}\geq \chi\cdot C_i$, it follows ~\eqref{eq:condition-2-re} and that
\begin{equation}
   \sum_{t\in[T]} g_{i,t}(\bar{v}_{i,t})\geq p_{\min}\cdot \chi\cdot C_i + \int_{\chi\cdot C_i}^{w_{i,T}}\phi_i(w)dw
\end{equation}
by the fact that if $\bar{v}_{i,t}>0$,
\begin{equation}
    g'_{i,t}(v)\geq max \{p_{\min}, \phi_i(w_{i,{t-1}}+v)\},\forall v\leq\bar{v}_{i,t}
    \label{eq:gradient}
\end{equation}
~\eqref{eq:gradient} is due to that $g'_{i,t}(\bar{v}_{i,t})\geq max \{p_{\min}, \phi_i(w_{i,t})\}$, which follows the concavity of $g_{i,t}(\cdot)$, non-decreasing property of $\phi_i(\cdot)$, and $g'_{i,t}(\bar{v}_{i,t})\geq \phi_i(w_{i,t})$, if $\bar{v}_{i,t}>0$ by the KKT condition of \textsf{(P\&D)} in~\eqref{eq:primal-dual-algorithm}.

\section{Proof of Lemma~\ref{thm:upper-bound-class-I}} 
\label{proof:upper-bound}

We first consider a new class of revenue function,
\begin{itemize}
    \item $g_t(v_t)$ is concave, increasing and differentiable with $g_t(0)=0$;
    \item $\frac{g'(0)}{g_t(\delta_t)/\delta_t}\leq \xi$;
    \item $g'_t(0)\in [p_{\min},p_{\max}]$
\end{itemize}
where $\xi$ is a given parameter. We denote this class of revenue function as \textsf{Class$_{\xi}$}. 
 We can generalize the results (Theorem 8) in~\cite{cr_pursuit2018} by taking rate limit into consider and obtain the following proposition
\begin{prop}
\label{thm:upper-bound-general} For \textsf{Class$_{\xi}$} revenue function, we have
\begin{equation}
    \Phi_{\xi}(\pi)\leq \xi\cdot( \ln\theta+1)\cdot\frac{C}{\pi}
\end{equation}
\end{prop}
We omit the detailed proof here as it is by simply checking the proof in~\cite{cr_pursuit2018} step-by-step when considering revenue function $g_t(v_t)$ defined over $v_t\in[0,\delta_t]$ instead of $v_t\in[0,C]$. 

We now turn to the proof of Lemma~\ref{thm:upper-bound-class-I}. 

 

\begin{proof}[Proof of  Lemma~\ref{thm:upper-bound-class-I}]

We proof the lemma by showing that for any $\epsilon>0$, $\Phi_{1}(\pi)\leq (1+\epsilon)(\ln\theta+1)\cdot\frac{C}{\pi}$. 

For any function $g_t(v_t),v_t\in[0,\delta_t]$, we can construct a sequence of functions as follows. We begin by finding the maximum $v^{(1)}\leq\delta_t$ such that $g'_t(v^{(1)})\geq g_t'(0)/(1+\epsilon)$. We define $g^{(1)}_t(v)=g_t(v),v\in[0,v^{(1)}]$. We then find the maximum $v^{(2)}\leq \delta_t$ such that $g'_t(v^{(2)})\geq g_t'(v^{(1)})/(1+\epsilon)$. We define  $g^{(2)}_t(v)=g_t(v^{(1)}+v)-g_t(v^{(1)}),v\in[0,v^{(2)}-v^{(1)}]$. We continue the steps until we arrive at $\delta_t$. Suppose in total there are $k_t$ functions we construct, and they are $\{g^{(i)}(v)\}_{i\in k_t}$, where $g^{(i)}(v)=g_t(v^{(i-1)}+v)-g_t(v^{(i-1)}), v\in[0,v^{(i)}-v^{(i-1)}]$. Also, $v^{(k)}=\delta_t$. We can easily check that $\{g^{(i)}(v)\}_{i\in k_t}$ belongs to  \textsf{Class$_{\xi=1+\epsilon}$}. 

For any $\sigma\in\Sigma$, suppose there is a revenue function $g_t(v_t)$ not belonging to \textsf{Class$_{\xi=1+\epsilon}$}, we can construct $\{g^{(i)}(v)\}_{i\in k_t}$ following the above procedure and replace  $g_t(v_t)$ in $\sigma$. We denote the new input as $\tilde{\sigma}$. We can show that the replacement will not decreasing the total allocation of \crp,  We note that the output of \crp~does not change at $\tau\neq t$ as the venue function and increment of the optimal objective remains the same at those slots. Thus it is sufficient to show that
\begin{equation}
  \hat{v}_t  \leq  \sum_{i=1}^{k_t}{\tilde{v}_t^{(i)}} \label{eq:show-increment}
\end{equation}
, where $\hat{v}_t$ is the output of \crp~ under $\sigma$ at slot $t$ and $\tilde{v}_t^{(i)}$ is the output of \crp~ under $\tilde{\sigma}$ at function $g_t^{(i)}(\cdot)$, for $i\in[k_t]$. Following the \crp~algorithm, we have
\begin{equation}
    g_t(\hat{v}_t) = \frac{1}{\pi}(OPT(t)-OPT(t-1)) =  \sum_{i=1}^{k_t} g_t^{(i)}({\tilde{v}_t^{(i)}})
\end{equation}
According to our construction and concavity of $g_t(\cdot)$, we have 
\begin{align}
    \sum_{i=1}^{k_t} g_t^{(i)}({\tilde{v}_t^{(i)}}) & = \sum_{i=1}^{k_t} \left[ g_t(v^{(i-1)}+{\tilde{v}_t^{(i)}})-g_t(v^{(i-1)}) \right] \\
    & \stackrel{(a)}{\leq} \sum_{i=1}^{k_t} \left[ g_t( \sum_{j=1}^{i-1}{\tilde{v}_t^{(j)}}+{\tilde{v}_t^{(i)}})-g_t( \sum_{j=1}^{i-1}{\tilde{v}_t^{(j)}})\right] \\
    & = g_t( \sum_{i=1}^{k_t}{\tilde{v}_t^{(i)}}) 
\end{align}
, where (a) is due to the concavity of $g_t(\cdot)$ and the fact that $\sum_{j=1}^{i-1}{\tilde{v}_t^{(j)}}\leq \sum_{j=1}^{i-1}{({v}^{(j)}-{v}^{(j-1)})}\leq v^{(i-1)}, \forall i\in[t_k]$. Thus, we have $ g_t(\hat{v}_t) \leq  g_t( \sum_{i=1}^{k_t}{\tilde{v}_t^{(i)}})  $ and conclude~\eqref{eq:show-increment}. 

This directly implies that we can replace all functions by a sequence of \textsf{Class$_{\xi=1+\epsilon}$} functions without decreasing the total allocation of \crp. Thus, $\Phi(\pi)\leq \Phi_{\xi=1+\epsilon}(\pi)\leq (1+\epsilon)\cdot (\ln\theta+1)\cdot\frac{C}{\pi},\forall \epsilon>0$. And, we conclude $\Phi(\pi)\leq(\ln\theta+1)\cdot\frac{C}{\pi}$



\end{proof}

\section{Proof of Theorem~\ref{thm:step-I}}
\label{app:thm:step-I}





We first show a useful proposition.
\begin{prop}
\label{thm:psi_increment}
$\Psi_{i,t}(a)$ is non-decreasing in $a$.
\end{prop}

\begin{proof}

By integration by parts,
\begin{align}
\int_0^{c_i} f_i(x) \frac{\partial G_{i,t}(x,a)}{\partial x} dx
=& \int_0^{c_i} f_i(x) dG_{i,t}\left(x,a\right) \\
=&  \left.f_i(x)\cdot G_{i,t}\left(x,a\right)\right|_{0}^{C_i}-\frac{1}{\pi\cdot C_i}\int_0^{C_i} f_i(x) G_{i,t}\left(x,a\right)dx\\
 =&f_i(C_i) G_{i,t}\left(C_i,a\right)-\frac{1}{\pi\cdot C_i}\int_0^{C_i} f_i(x) G_{i,t}\left(x,a\right)dx\\
=  &\Psi_{i,t}(a) \label{eq:alter_psi}
\end{align}

And according to the sensitivity analysis of the optimization problem defining $G_{i,t}(x,a)$, we have 
\begin{equation}
    \frac{\partial G_{i,t}(x,a)}{\partial x} = \eta^*.
\end{equation}
where  $\eta^* $ is the optimal dual variable associated with constraint~\eqref{eq:dfn_G_x}. We can check that $\eta^*$ is non-decreasing in $a$ by the KKT condition, and thus $\Psi_{i,t}(a)$ is non-decreasing in $a$.
\end{proof}

\begin{proof}[Proof of Theorem~\ref{thm:step-I}]
We adapt the primal-and-dual framework~\cite{buchbinder2009design,Feldman2009Online} to prove the theorem. 
We begin with the revenue increment of our algorithm \textsf{\allallt($\pi$)-A}~at slot $t$, denoted as $\Delta P$. According to~ \eqref{eq:defn_G-C},
\begin{equation}
    \Delta P \triangleq \sum_{i} \left( \tilde{OPT}_{i,t} - \tilde{OPT}_{i,t-1} \right)= \sum_{i}\left(G_{i,t}(C_i,\hat{a}_{i,t})-G_{i,t}(C_i,0)\right).\label{eq:primal-increment-C}
\end{equation}

We note that the optimal solution of \textsf{\allallt($\pi$)-A}  satisfies the KKT condition,
\begin{align}
     & \tilde{g}'_{i,t}(\hat{a}_{i,t})-\Psi_{i,t}(\hat{a}_{i,t})-\tilde{\beta}_t-\tilde{\gamma}_{i,t}+\tilde{\sigma}_{i,t} =0,\forall i\label{eq:kkt-aat-gradient-C}\\
    & \tilde{\beta}_t\left(\sum_i \hat{a}_{i,t}-\pi\cdot A_t\right) =0,\label{eq:kkt-aat-beta-C}\\
    & \tilde{\sigma}_{i,t}\cdot \hat{a}_{i,t}=0, \forall i,\label{eq:kkt-aat-sigma-C}\\
    & \tilde{\gamma}_{i,t}\left(\hat{a}_{i,t}- \pi \cdot \delta_{i,t}\right) = 0, \forall i,\label{eq:kkt-aat-gamma-C}
\end{align}
, where $\tilde{\beta}_t\geq 0$ and $\tilde{\sigma}_t^i\geq 0,\tilde{\gamma}_t^i\geq 0,\forall i$ are dual variables corresponding to constraints \eqref{eq:AAt-allowance-C} and \eqref{eq:AAt-rate}). 

We first write down the dual problem of \probname~at slot $t$,
\begin{align}
    \textsf{Dual-\probname}:\;\quad & \min \sum_{i,\tau\in[t]}h_{i,\tau}(\alpha_i+\beta_\tau) + \sum_{i}{C_i\alpha_i}+\sum_{\tau}{A_{\tau}\beta_{\tau}}\\
    \text{s.t.} \quad& \alpha_i\geq 0, \beta_t\geq 0, \forall t,i, \label{eq:dual-owfm-var-C}
\end{align}
where 
\begin{equation}
h_{i,\tau}(\lambda)=\max_{0\leq v\leq \delta_{i,\tau}} g_{i,\tau}(v)-\lambda\cdot v.
\end{equation}

We compare our online increment with the following dual solution. At slot $t$, we update the dual variable $\{\alpha_{i,t}\}_{i\in[N]}$, determine the dual variable $\beta_t$,
\begin{equation}
    \alpha_{i,t}  = \Psi_{i,t}(\hat{a}_{i,t}),\forall i; \beta_t  = \tilde{\beta}_t.
\label{eq:dual-baseline-owfm-C}
\end{equation}
We note that the dual variable satisfies the dual constraint~\eqref{eq:dual-owfm-var-C}.

Let the increment of the dual objective by above dual solutions at each slot as $\Delta D$. To proof the theorem, the most important step in the framework is to show that at each slot, we have  
\begin{equation}
    \Delta D \leq \frac{1}{\pi}\frac{e^\frac{1}{\pi}}{e^{\frac{1}{\pi}}-1} {\Delta P} 
    \label{eq:pd-key-C}
\end{equation}

We can now compute the increment of the dual,
\begin{align}
    \Delta D & =  \sum_{i,\tau<t}\left( h_{i,\tau}\left(\alpha_{i,t}+{\beta}_{\tau}\right)- h_{i,\tau}\left(\alpha_{i,t-1}+{\beta}_{\tau}\right)\right) \nonumber\\
    & \quad   + \sum_{i} h_{i,t}\left(\alpha_{i,t}+{\beta}_t\right) + \sum_{i} C_i \left(\alpha_{i,t}-\alpha_{i,t-1}\right) + \beta_t A_t\\
    &  \stackrel{(a)}{\leq}  \sum_{i} h_{i,t}\left(\Psi_{i,t}(\hat{a}_{i,t})+\tilde{\beta}_t\right) + \sum_{i} C_i \left(\Psi_{i,t}(\hat{v}_{i,t})-\Psi_{i,t}(0)\right) + \tilde{\beta}_t A_t  \\
    &  \stackrel{(b)}{=} \sum_{i} \frac{1}{\pi}\left[\tilde{g}_{i,t}(\hat{a}_{i,t})-\left( \Psi_{i,t}(\hat{a}_{i,t})+\tilde{\beta}_t\right)\hat{a}_{i,t}\right] + \sum_{i} C_i \left(\Psi_{i,t}(\hat{a}_{i,t})-\Psi_{i,t}(0)\right) + \tilde{\beta}_t A_t \\
    & \stackrel{(c)}{=} \sum_{i} C_i \left(\Psi_{i,t}\left(\hat{a}_{i,t}\right)-\Psi_{i,t}(0)\right) + \frac{1}{\pi}\sum_i \left(\tilde{g}_{i,t}(\hat{a}_{i,t}) - \Psi_{i,t}(\hat{a}_{i,t})\cdot \hat{a}_{i,t}\right)
\end{align}
We show the above equality of inequality one-by-one. 
\begin{itemize}
    \item 
(a) is according to~\eqref{eq:dual-baseline-owfm-C} (the way to set the dual variables) and the facts that $a_{i,t}\geq a_{i,t-1}$ and $h_{i,t}(\lambda)$ is non-increasing in $\lambda$. 
\item
(b) is according to the fact that by~\eqref{eq:kkt-aat-gradient-C},~\eqref{eq:kkt-aat-sigma-C} and~\eqref{eq:kkt-aat-gamma-C}$, \hat{a}_{i,t}$ is the optimal solution to 
\begin{equation}
    \max_{0\leq v\leq \pi\cdot \delta_{i,\tau}} \tilde{g}_{i,\tau}(v)-\left( \Psi_{i,t}(\hat{a}_{i,t})+\tilde{\beta}_t\right)\cdot v.
\end{equation}
Also, we have
\begin{align}
   & \max_{0\leq v\leq \pi\cdot \delta_{i,\tau}} \tilde{g}_{i,\tau}(v)-\left( \Psi_{i,t}(\hat{a}_{i,t})+\tilde{\beta}_t\right)\cdot v\\
   = &\max_{0\leq v\leq \pi\cdot \delta_{i,\tau}} \pi\cdot g_{i,\tau}(\frac{v}{\pi})-\left( \Psi_{i,t}(\hat{a}_{i,t})+\tilde{\beta}_t\right)\cdot v\\
   = & \pi\cdot \max_{0\leq v\leq \pi\cdot \delta_{i,\tau}} \pi\cdot g_{i,\tau}(\frac{v}{\pi})-\left( \Psi_{i,t}(\hat{a}_{i,t})+\tilde{\beta}_t\right)\cdot \frac{v}{\pi}\\
   = & \pi\cdot \max_{0\leq v\leq \delta_{i,\tau}} \pi\cdot g_{i,\tau}(v)-\left( \Psi_{i,t}(\hat{a}_{i,t})+\tilde{\beta}_t\right)\cdot v \\
   = & \pi \cdot h_{i,t} \left(\Psi_{i,t}(\hat{a}_{i,t})+\tilde{\beta}_t\right).
\end{align}
\item
(c) is according to~\eqref{eq:kkt-aat-beta-C}.
\end{itemize}


Recall \eqref{eq:defn_Phi-C} that 
\[
\Psi_{i,t}(a) = f_i(C_i)\cdot G_{i,t}(C_i,a)-\frac{1}{\pi\cdot C_i}\int_0^{C_i} G_{i,t}\left(x,a\right)\cdot f_i(x) dx.
\]

Plugging $\Psi_{i,t}(\cdot)$, i.e., \eqref{eq:defn_Phi-C}), in $\Delta D$, we have that 
\begin{align}
    \Delta D   & =\sum_i\left\{ C_i\cdot f_i(C_i) G_{i,t}\left(C_i,\hat{a}_{i,t}\right)-\frac{1}{\pi}\cdot\int_0^{C_i} f_i(x) G_{i,t}\left(x,\hat{a}_{i,t}\right)dx\right.\\
    &\left.-\left(   f_i(C_i) G_{i,t}\left(C_i,0\right)- \frac{1}{\pi}\cdot\int_0^{C_i} f_i(x) G_{i,t}\left(x,0\right)dx\right) +\frac{1}{\pi}\cdot\left(\tilde{g}_{i,t}(\hat{a}_{i,t}) - \Psi_{i,t}(\hat{a}_{i,t})\cdot \hat{a}_{i,t}\right)\right\}
\end{align}

By the fact that 
$C_i\cdot f_i(C_i)=\frac{1}{\pi}\frac{e^{\frac{1}{\pi}}}{e^{\frac{1}{\pi}}-1}$ according to ~\eqref{eq:defn_f-C}, we have 
\begin{align}
\Delta D =  &  \frac{1}{\pi}\frac{e^{\frac{1}{\pi}}}{e^{\frac{1}{\pi}}-1} \sum_i \left[G_{i,t}\left(C_i,\hat{a}_{i,t}\right)- G_{i,t}\left(C_i,0\right)\right]\nonumber\\
    & + \sum_{i} \left[-\frac{1}{\pi}\cdot\int_0^{C_i} f_i(x) G_{i,t}\left(x,\hat{a}_{i,t}\right)dx+ \frac{1}{\pi}\cdot\int_0^{C_i} f_i(x) G_{i,t}\left(x,0\right)dx \right]\nonumber\\
    & + \sum_{i}\frac{1}{\pi}\cdot \left(\tilde{g}_{i,t}(\hat{a}_{i,t}) - \Psi_{i,t}(\hat{a}_{i,t})\cdot \hat{a}_{i,t}\right) \label{eq:dual-increment2-C}
\end{align}

Comparing with~\eqref{eq:primal-increment-C}, to show \eqref{eq:pd-key-C}, is sufficient to show that
\begin{equation}
\tilde{g}_{i,t}(\hat{a}_{i,t}) - \Psi_{i,t}(\hat{a}_{i,t})\cdot \hat{a}_{i,t}\leq \int_0^{C_i} f_i(x) G_{i,t}\left(x,\hat{a}_{i,t}\right)dx -  \int_0^{C_i} f_i(x) G_{i,t}\left(x,0\right)dx \label{eq:optimal-objective-bound-C}
\end{equation}
 
We further have 
\begin{align}
   & \tilde{g}_{i,t}(\hat{a}_{i,t}) - \Psi_{i,t}(\hat{a}_{i,t})\cdot \hat{a}_{i,t}\\
    & \leq \tilde{g}_{i,t}(\hat{a}_{i,t})-\int_{0}^{\hat{a}_{i,t}}\Psi_{i,t}(a)da\\
   & \leq \tilde{g}_{i,t}(\hat{a}_{i,t})-\int_{0}^{\hat{a}_{i,t}}\left[f_i(C_i)\cdot G_{i,t}(C_i,a)-\frac{1}{C_i}\int_0^{C_i} G_{i,t}\left(x,a\right)\cdot f_i(x) dx\right]da\\
   & \stackrel{(a)}{=} \tilde{g}_{i,t}(\hat{a}_{i,t})-\int_{0}^{\hat{a}_{i,t}}\int_0^{c_i}\frac{\partial G_{i,t}(x,a)}{\partial x} \cdot f_i(x) dxda\\
   &= \int_0^{c_i}\int_{0}^{\hat{a}_{i,t}} \tilde{g}'_{i,t}(a)- \frac{\partial G_{i,t}(x,a)}{\partial x} da\cdot f_i(x) dx \label{eq:optimal-objective-bound2-C}
\end{align}
, where (a) is due to ~\eqref{eq:alter_psi}.

So, it reduces to show that
\begin{equation}
\label{eq:improvement_bound-C}
    \int_{0}^{\hat{a}_{i,t}} \tilde{g}'_{i,t}(a)- \frac{\partial G_{i,t}(x,a)}{\partial x} da\leq G_{i,t}\left(x,\hat{a}_{i,t}\right)-G_{i,t}\left(x,0\right)
\end{equation}
, which is equivalent to

\begin{equation}
\label{eq:improvement_bound2-C}
    \int_{0}^{\hat{a}_{i,t}}  \tilde{g}'_{i,t}(a)- \frac{\partial G_{i,t}(x,a)}{\partial x} da\leq \int_{0}^{\hat{a}_{i,t}} \frac{\partial G_{i,t}(x,a)}{\partial a} da.
\end{equation}

To show the above inequality, it is sufficient to show that,
\begin{equation}
\label{eq:improvement_rate_bound-C}
    \tilde{g}'_{i,t}(a)\leq \frac{\partial G_{i,t}(x,a)}{\partial x} +
\frac{\partial G_{i,t}(x,a)}{\partial a}.
\end{equation}
To proceed, we recall that $G_{i,t}(x,a)$ is the optimal objective to the following problem,

\begin{align}
G_{i,t}\left(x,a\right) = \max\quad & \sum_{\tau\in[t]}\tilde{g}_{i,\tau}( v_{i,\tau}) \label{eq:defn_G_re-C}\\
\text{s.t}\quad  & \sum_{\tau\in[t]} v_{i,\tau}\leq x\label{eq:defn_G_inv_re-C}\\
& 0\leq v_{i,t}\leq a\label{eq:defn_G_rate_t_re-C}\\
 &0\leq v_{i,\tau}\leq \hat{a}_{i,\tau}, \forall \tau\in [t-1].\label{eq:defn_G_rate_tau_re-C}
\end{align}

Let $\eta$, $\psi_t$ and $\phi_t$, and $\{\psi_\tau\}_{\tau\in[t-1]}$ and $\{\phi_\tau\}_{\tau\in[t-1]}$ be the dual variable associated with \eqref{eq:defn_G_inv_re-C},~\eqref{eq:defn_G_rate_t_re-C}, and~\eqref{eq:defn_G_rate_tau_re-C}, respectively.


According to sensitivity analysis of convex program, we have 
\begin{equation}
\label{eq:dual-sensitivity-C}
\frac{\partial G_{i,t}(x,a)}{\partial x} = \eta^*,
\frac{\partial G_{i,t}(x,a)}{\partial a}=\phi_t^*
\end{equation}
According to the KKT condition of the optimal solution, we have 
\begin{equation}
\label{eq:improvement_rate_bound2-C}
    \tilde{g}'_{i,t}(a)\leq  \tilde{g}'_{i,t}(v_{i,t}^*)=\eta^* + \phi_t^*-\psi_t^*\leq \eta^* + \phi_t^*
\end{equation}
, where $v_{i,t}^*$, $ \phi_t^*$, $\psi_t^*$, $\eta^*$ and $\phi_t^*$ represent the optimal primal solution and dual solution, and the first inequality is due to the fact that $ \tilde{g}_{i,t}(\cdot)$ is concave and $v_{i,t}^*\leq a$.

Combining~\eqref{eq:dual-sensitivity-C} and~\eqref{eq:improvement_rate_bound2-C}, we conclude~\eqref{eq:improvement_rate_bound-C}, which leads to~\eqref{eq:improvement_bound-C} and~\eqref{eq:improvement_bound2-C}. Combining~\eqref{eq:optimal-objective-bound2-C} and~\eqref{eq:improvement_bound-C}, we conclude~\eqref{eq:optimal-objective-bound-C}. Combining ~\eqref{eq:optimal-objective-bound-C},~\eqref{eq:primal-increment-C} and ~\eqref{eq:dual-increment2-C}, we easily conclude~\eqref{eq:pd-key-C}. Summing~\eqref{eq:pd-key-C} over all time slot, we have 
\begin{equation}
 \sum_{i\in[N]} \tilde{OPT}_{i,t} \geq  \pi\cdot \frac{e^{\frac{1}{\pi}}-1}{e^\frac{1}{\pi}}\cdot \text{Dual-OOIC-M}\geq \pi\cdot \frac{e^{\frac{1}{\pi}}-1}{e^\frac{1}{\pi}}\cdot OPT_t
\end{equation}
\end{proof}

\end{document}